\documentclass[12pt, amsfonts, epsfig, reqno, a4paper]{amsart}
\usepackage[hyperindex,breaklinks]{hyperref}
\usepackage{amsmath}
\usepackage{amsfonts}
\usepackage{mathrsfs}
\usepackage{graphicx}
\usepackage{color}
\usepackage{slashed}
\usepackage{braket}
\usepackage{extarrows}
\usepackage{amsmath, bm}
\usepackage{amssymb}
\usepackage[margin=1.5in]{geometry}
\usepackage{bbm}

\allowdisplaybreaks[4]
\usepackage{hyperref}

\newcommand{\C}{{\mathbb{C}}}
\newcommand{\Q}{{\mathbb{Q}}}

\newcommand{\Z}{{\mathbb{Z}}}

\newcommand{\half}{\frac12}

\newcommand{\N}{{\mathbb{N}}}

\definecolor{yellow}{rgb}{1,.7,0}

\definecolor{pkured}{rgb}{0.55,0,0}

\newcommand{\be}{\begin{equation*}}
	\newcommand{\ee}{\end{equation*}}
\newcommand{\beq}{\begin{equation}}
	\newcommand{\eeq}{\end{equation}}

\numberwithin{equation}{section}

\newtheorem{cor}{Corollary}[section]

\newtheorem{lem}[cor]{Lemma}
\newtheorem{prop}[cor]{Proposition}
\newtheorem{thm}[cor]{Theorem}

\newtheorem{defn}[cor]{Definition}

\newtheorem{ex}[cor]{Example}
\newtheorem{rmk}[cor]{Remark}

\numberwithin{figure}{section}
\usepackage{ytableau}

\usepackage{tikz}
\newcounter{x}
\newcounter{y}
\newcounter{z}

\author{Zhiyong Wang}
\email{zwangco@connect.ust.hk}
\address{Zhiyong Wang, School of Mathematics and Statistics, Wuhan University, Wuhan}

\author{Chenglang Yang}
\email{yangcl@pku.edu.cn}
\address{Chenglang Yang, Institute for Math and AI, Wuhan University, Wuhan}

\title[Correlation Function of Self-Conjugate Partitions]
{Correlation Function of Self-Conjugate Partitions:
$q$-Difference Equation and Quasimodularity}

\begin{document}
	\maketitle
	%%%%%%%%%%%%%%%%%%%%%%%%%%%%%%%%%%%%%%%%%%%%%%%%%%%%%%%%%%%%%%%%%%%%%%%%%%%%%%%%%%%

\begin{abstract}
	In this paper, we introduce the $\omega$-transform on the fermionic operators and apply it to study the uniform measure for the self-conjugate partitions.
	We first derive the $q$-difference equation which is satisfied by the $n$-point correlation function related to the uniform measure.
	As applications,
	we give explicit formulas for the one-point and two-point functions via Eisenstein series.
	Motivated by this, we prove the quasimodularity of the general $n$-point function.
	We also derive the limit shape of self-conjugate partitions under the Gibbs uniform measure and compare it to the leading asymptotics of the one-point function.
\end{abstract}

\section{Introduction}

The integer partitions are intensively studied by mathematicians,
including their relation to combinatorics,
representation theory,
number theory,
random geometry,
and mathematical physics
(see, for examples, \cite{And98,Mac,NO06,O05} and reference therein).
In 2000,
Bloch and Okounkov \cite{BO} studied the characters of the infinite wedge representation which are exhibited as certain correlation functions on the set of all integer partitions.
They derived $q$-difference equations for their correlation functions,
and obtained explicit formulas for correlation functions in terms of the theta functions and their derivatives.
Their result reveals a deep connection between correlation functions of partitions and quasimodularity (see also \cite{Zhu96}).
Special cases of their correlation functions were also studied by Dijkgraaf \cite{Dij95} earlier from the viewpoint of mirror symmetry for elliptic curves.
The explicit formulas and quasimodularity for Bloch-Okounkov's correlation functions were also proved by Zagier using a purely combinatorial method later \cite{Z16}.
Moreover,
Bloch and Okounkov's result and its generalizations have great applications in many areas including the limit behavior of random partitions \cite{O01,O05},
Gromov-Witten theory of elliptic curve \cite{Eng21,OP06},
intersection numbers on Hilbert schemes of points \cite{LQW},
and moduli spaces of Abelian differentials \cite{CMZ18,EO01,GM20}, etc.

In this paper,
we initiate the study of correlation functions of the self-conjugate partitions.
More precisely,
denote the set of all self-conjugate partitions as $\mathscr{P}^s$.
We are mainly interested in the following $n$-point function
\begin{align}\label{eqn:def G main}
	G(t_1,t_2,...,t_n)
	=\frac{1}{\sum\limits_{\lambda\in\mathscr{P}^s}q^{|\lambda|}}
	\cdot \sum\limits_{\lambda\in\mathscr{P}^s}\prod\limits_{j=1}^{n}\sum\limits_{i=1}^{\infty} t_j^{\lambda_i-i+\frac{1}{2}}q^{|\lambda|}.
\end{align}
From a probabilistic viewpoint, the $n$-point function above $G(t_1,t_2,...,t_n)$ can be regarded as a kind of moment generating function of the Gibbs uniform measure on the set of self-conjugate partitions.
It is well-known that all integer partitions label the basis of the charge zero infinite wedge space.
Thus, the correlation functions studied by Bloch and Okounkov can be represented as a trace on the infinite wedge space (see \cite{BO,O01}).
Then, the standard trace properties can be directly applied in their study.
For another example,
Wang \cite{W04} used a similar method to study the correlation functions of strict partitions,
and in this case, the set of strict partitions labels a basis of the twisted Fock space
(see also \cite{TW07}).
Thus, a difficulty in studying the correlation functions of self-conjugate
partitions arises because they only form a subset of a basis.
The main method to conquer the difficulty in this paper is the $\omega$-transform on the fermionic operators and the fermionic Fock space,
which is introduced in Section \ref{sec:main result}.
We will use the $\omega$-transform to study the $n$-point function $G(t_1,t_2,...,t_n)$ of the self-conjugate partitions defined in equation \eqref{eqn:def G main}.
We first derive the following $q$-difference equation.

\begin{thm}\label{main} The $n$-point function $G(t_1,t_2,...,t_n)$ satisfies the following $q$-difference equation
		\begin{align}\label{eqn:q-diff main}
			\begin{split}
			G(&q^{-2}t_1,t_2,\ldots,t_n)\\
			&=\sum\limits_{P^-,P^+ \subseteq \{2\ldots,n\}, \atop P^-\cap P^+=\emptyset}(-1)^{|P^-|-1}G\Big(t_1 \prod\limits_{i\in P^-}t_i^{-1}
			\cdot\prod_{j\in P^+}t_j,\cdots,\hat{t_i},\cdots,\hat{t_j},\cdots\Big),
			\end{split}
		\end{align}
	where the notation $\hat{\cdot}$ means that the corresponding term should be omitted.
	Parallel formula for $G(t_1,\ldots,q^{-2}t_j,\ldots,t_n)$ is achieved by exchanging the variables $t_1,t_2,\dots,t_n$.
\end{thm}

Following the spirit of Bloch and Okounkov \cite{BO},
the $q$-difference equation \eqref{eqn:q-diff main}, together with the analysis of singularities, uniquely specifies all these $n$-point functions.
We use this method to obtain explicit formulas for the one-point and two-point functions,
which helps us to establish the quasimodularity for these cases.
\begin{cor}
	\label{cor:main one}
	The one-point function $G(t)$ is given by
	\begin{align}\label{eqn:G1 main}
		G(t)&=q^{1/4}\prod_{m=1}^{\infty}\frac{(1-q^{2m})^2}{(1+q^{2m-1})^2}\cdot\frac{\Theta_3(t;q)}{\Theta_1(t;q)},
	\end{align}
	where $\Theta_1(t;q):=\sum_{n\in \Z}(-1)^nq^{(n+\frac{1}{2})^2}t^{n+\frac{1}{2}}$ and
	$\Theta_3(t;q):=\sum_{n\in \Z}q^{n^2}t^{n}$ are the classic theta functions.
	Moreover,
	by virtue of the Eisenstein series,
	we have
	\begin{align}\label{eqn:G as Gl}
		G(t)=\frac{1}{2\pi i z}\exp\Bigg(\sum\limits_{\ell\in 2\Z_+} 2\left(G_\ell(\tau)- G^{(1,1)}_\ell(\tau)\right)\frac{(2\pi i z)^\ell}{\ell!}\Bigg),
	\end{align}
	where $q=e^{\pi i \tau}, t=e^{2\pi i z}$,
	$G_\ell(\tau)$ is the standard weight $\ell$ Eisenstein series for $\Gamma(1)=SL_2(\mathbb{Z})$,
	and $G^{(1,1)}_\ell(\tau)$ is the weight $\ell$ Eisenstein series for the congruence subgroup $\Gamma(2)$ with index vector $(1,1) \in \mathbb{Z}_2\times\mathbb{Z}_2$.
	
	The two-point function $G(t_1,t_2)$ is given by
	\begin{align}\label{eqn:G2 main}
		\begin{split}
		G(t_1,t_2)&=q^{1/4}\prod_{m=1}^{\infty}\frac{(1-q^{2m})^2}{(1+q^{2m-1})^2}
		\cdot\bigg[\frac{\Theta^{'}_3(t_1t_2;q)}{\Theta_1(t_1t_2;q)}-\\
		&\qquad\frac{\Theta^{'}_1(t_1;q)}{\Theta_1(t_1;q)}\cdot\frac{\Theta_3(t_1/t_2;q)}{\Theta_1(t_1/t_2;q)}-\frac{\Theta^{'}_1(t_2;q)}{\Theta_1(t_2;q)}\cdot\frac{\Theta_3(t_2/t_1;q)}{\Theta_1(t_2/t_1;q)}\bigg],
		\end{split}
	\end{align}
	where $\Theta^{'}_1(t;q):=t\frac{\partial}{\partial t}\Theta_1(t;q)$ and $\Theta^{'}_3(t;q)$ is defined similarly.
	See Section \ref{sec:app} for more details.
\end{cor}

The quasimodularity of Bloch and Okounkov's correlation functions on the set of all integer partitions \cite{BO} directly follows from their analysis on characters of the infinite wedge representation,
as well as their explicit formulas in terms of theta function.
For special cases, see also \cite{Dij95} and \cite{KZ95}.
A purely combinatorial proof of the quasimodularity was obtained by Zagier \cite{Z16},
who also pointed out that this quasimodularity should hold for more functions on the set of all integer partitions.
For more details and generalizations, see also \cite{CMZ18,EOP,GM20,HIL,I21,TW07} and references therein.
Motivated by the results above,
and Corollary \ref{cor:main one} for explicit formulas of the one-point and two-point functions,
we prove the following theorem about the quasimodularity for the general $n$-point function $G(t_1,t_2,\cdots,t_n)$ of self-conjugate partitions studied in this paper.
\begin{thm}\label{thm:main quasimod}
	Let $t_i=e^{2\pi i z_i}, i=1,2,\cdots,n$,
	and $q=e^{\pi i \tau}$.
	Expand the $n$-point function $G(t_1,t_2,\cdots,t_n)$ for the self-conjugate partitions as
	\begin{align}
		G(t_1,t_2,\cdots,t_n)
		=\sum\limits_{\ell_1,\ell_2,\cdots,\ell_n\geq0}\langle Q_{\ell_1}Q_{\ell_2}\cdots Q_{\ell_n} \rangle^s_q
		\cdot \prod_{j=1}^n (2\pi i z_j)^{\ell_j-1}.
	\end{align}
	Then, for any non-negative integers $\ell_1,...,\ell_n$,
	$$\langle Q_{\ell_1}Q_{\ell_2}\cdots Q_{\ell_n} \rangle^s_q$$
	is a quasimodular form of weight $\sum_{j=1}^n \ell_j$ for the congruence subgroup $\Gamma(2)$.
\end{thm}

We consider another generating series for the $q$-bracket $\langle\cdot\rangle^s_q$ of shifted symmetric functions.
The expansion coefficients of this generating series can be explicitly represented in terms of Eisenstein series (see equations \eqref{eqn:eisenl=1} and \eqref{eqn:eisen}),
which proves Theorem \ref{thm:main quasimod}.
The cases of $n=1$ and $n=2$ of the theorem above also directly follow from the explicit formulas in Corollary \ref{cor:main one} (see Remark \ref{rek:onetwo} for more details).
In general,
the $n$-point function is related to certain moments of Gibbs uniform measure on the set of self-conjugate partitions,
thus the quasimodularity should be connected to the shape fluctuations of self-conjugate partitions (see the discussions in Appendix A.2 of \cite{EO01}).

The study of the limit behavior of large partitions has a long history.
For instance,
Erd\"os and Lehner obtained the distribution of the largest part of a large partition under the uniform measure \cite{EL}.
Vershik proved the famous limit shape theorem for several measures on partitions \cite{V96}.
For more applications of random partitions,
see also \cite{BDJ,BOO,FVY,F93,O01,O05,TW} and reference therein.
In this paper,
we derive the limit shape of self-conjugate partitions under the Gibbs uniform measure.
\begin{prop}\label{prop:limit shape}
	When $q$ goes to $1^-$,
	the limit shape of the rescaled Young diagram of self-conjugate partitions under the measure $\mathfrak{M}_q(\cdot)$ is described by the graph of the following function
	\begin{align}\label{eqn:limit shape main}
		f(x)=\frac{\sqrt{6}}{\pi}\log\big(1-\exp(-\pi x/\sqrt{6})\big).
	\end{align}
	More precisely,
	if we use the graph of the function $f_{\lambda}(x)$ to represent the Young diagram of $\lambda$
	and denote its rescaled version by
	\begin{align*}
		\tilde{f}_\lambda(x):=
		4\sqrt{6}r \cdot f_\lambda(x/4\sqrt{6}r),
	\end{align*}
	where $r=-\frac{1}{2\pi}\cdot \log q$.
	Then,
	for any fixed $x>0$ and $\epsilon>0$,
	we have the following limit
	\begin{align*}
		\lim_{q\rightarrow1^-}\ 
		\mathfrak{M}_q\big(\big\{\lambda\big|\ |\tilde{f}_{\lambda}(x)-f(x)|<\epsilon\big\}\big)=1.
	\end{align*}
\end{prop}
We also verify that the leading asymptotics of the one-point function $G(t)$ given by equation \eqref{eqn:G as Gl} is compatible with the typical self-conjugate partition described by the limit shape in equation \eqref{eqn:limit shape main},
as what Eskin and Okounkov did in the Appendix of \cite{EO01}.
See Corollary \ref{cor:aysm} for more details.

The rest of this paper is organized as follows.
In Section \ref{sec:pre},
we review the notions of partitions and fermions.
In Section \ref{sec:main result},
we introduce the $\omega$-transform on the fermionic operators and the fermionic Fock space.
By virtue of that,
we prove Theorem \ref{main},
which gives the $q$-difference equation satisfied by $n$-point function $G(t_1,t_2,...,t_n)$.
As applications,
we obtain the explicit formulas of the one-point and two-point functions in Section \ref{sec:app},
which proves Corollary \ref{cor:main one}.
In Section \ref{sec:quasimo}, we prove Theorem \ref{thm:main quasimod}, which establishes quasimodularity for the general $n$-point function.
Finally in Section \ref{sec:limit shape},
we derive the limit shape of self-conjugate partitions and prove Proposition \ref{prop:limit shape}.
We also verify the compatibility of the leading asymptotics of the one-point function and the limit shape in the same section.

\section{Preliminaries}
\label{sec:pre}
In this section,
we review the notions of partition and the fermionic Fock space.
We recommend the books \cite{And98,DJM,Mac} for interested readers.
\subsection{Partitions} A partition of a non-negative integer $n$ is a sequence 
\[\lambda=(\lambda_1,\lambda_2,\ldots,\lambda_l)\]
of positive integers satisfying the non-increasing condition
$\lambda_1\geq\lambda_2\geq\cdots\geq\lambda_l$
and
\[\sum\limits_{i=1}^{l}\lambda_i=n.\]
The length and size of this partition $\lambda$ are $l(\lambda)=l$ and $|\lambda|=n$ respectively.
For $i>l(\lambda)$,
we use the notation $\lambda_i=0$ for convenience.
Each partition is uniquely represented by its Young diagram.
The Young diagram of $\lambda$ has $\lambda_i$ boxes in the $i$-th row.
For example,
the Figure \ref{fig:yd} is the Young diagram of the partition $(8,4,4,2,1)$.
\begin{figure}[htbp]\label{fig:yd}
	\ydiagram{8,4,4,2,1}
	\caption{The Young diagram corresponding to $(8,4,4,2,1)$}
\end{figure}

The conjugation $\lambda^t$ of a partition $\lambda$ is obtained by reflection along the main diagonal of the Young diagram corresponding to this partition.
More precisely,
$\lambda^t$ is the partition of length $\lambda_1$ defined by
\begin{align*}
	\lambda^t_k:=\#\{i|\lambda_i\geq k\},
	\quad 1\leq k\leq \lambda_1.
\end{align*}
For example,
the conjugation of the partition $(8,4,4,2,1)$ in Figure \ref{fig:yd} is $(5,4,3,3,1,1,1,1)$.
The partition $\lambda$ is called self-conjugate if $\lambda=\lambda^t$.
Intuitively,
if $\lambda$ is self-conjugate,
then the Young diagram corresponding to $\lambda$ is invariant under reflection across the main diagonal. 
We denote by $\mathscr{P}$ and $\mathscr{P}^s$ the set of all partitions and the set of all self-conjugate partitions respectively.

The Frobenius notation of a partition $\lambda$ is defined by
\begin{align*}
	\lambda=(m_1,...,m_r|n_1,...,n_r),
\end{align*}
where $r$ is the length of the main diagonal of the Young diagram corresponding to $\lambda$ and
\begin{align*}
	m_i=\lambda_i-i,
	\qquad\qquad
	n_i=\lambda^t_i-i,
	\qquad 1\leq i \leq r.
\end{align*}
We call $r=r(\lambda)$ the Frobenius length and $(m_i|n_i)$ the Frobenius coordinates of this partition.
One can notice that,
a partition $\lambda$ is self-conjugate if and only if $m_i=n_i$ for all $i=1,...,r(\lambda)$.

\subsection{Uniform measure for self-conjugate partitions}
\label{sec:def m}
In this paper,
we shall study the following measure
\begin{align}\label{eqn:def measure}
	\mathfrak{M}_q(\lambda):=\frac{q^{|\lambda|}}
	{\prod_{i=1}^{\infty}(1+q^{2i-1})},
	\qquad q\in [0,1),
\end{align}
on the  set of self-conjugate partitions $\mathscr{P}^s$.
It is called the Gibbs uniform measure in \cite{FVY,V96}.
Notice that the probability of a partition under this measure $\mathfrak{M}_q(\cdot)$ only depends on the size of this partition.
Thus,
the restriction of this measure $\mathfrak{M}_q(\cdot)$ to the set $\mathscr{P}^s(n)$,
which consists of self-conjugate partitions of size $n$,
is exactly the usual uniform measure on the set $\mathscr{P}^s(n)$.

The normalization factor for the measure $\mathfrak{M}_q(\cdot)$ in equation \eqref{eqn:def measure} comes from
\begin{align}\label{eqn:Z factor}
	Z_s(q)
	:=\sum\limits_{\lambda\in\mathscr{P}^s}q^{|\lambda|}
	=\prod\limits_{i=1}^{\infty}(1+q^{2i-1}),
\end{align}
which makes $\mathfrak{M}_q(\cdot)$ a probability measure.
Moreover,
it is obvious that equation \eqref{eqn:Z factor} is an analytic function when $|q|<1$.
The equation \eqref{eqn:Z factor} follows from the fact that
a self-conjugate partition $\lambda$ is uniquely determined by its Frobenius coordinates $m_i=\lambda_i-i, i=1,...,r(\lambda)$.

For an arbitrary function $f:\mathscr{P}^s \rightarrow \mathbb{C}$,
we study the $q$-bracket of $f$ related to the measure $\mathfrak{M}_q(\cdot)$ as
$$\langle f \rangle_q^s:
=\sum\limits_{\lambda\in\mathscr{P}^s}f(\lambda) \mathfrak{M}_q(\lambda)
=\frac{\sum\limits_{\lambda\in\mathscr{P}^s}f(\lambda)q^{|\lambda|}}{\sum\limits_{\lambda\in\mathscr{P}^s}q^{|\lambda|}}
\in\mathbb{C}[\![q]\!].$$
Here $\langle f\rangle_q$ is regarded as a formal power series of $q$. 
For a large class of $f$,
this $q$-bracket is expected to exhibit interesting properties such as the analyticity and modularity (see, for examples, Section 9 in \cite{Z16}).

%We assume that $\langle f \rangle_q$ converges for all $q$ in the unit ball of $\mathbb{C}$ and is a holomorphic function in the complex upper half-plane $\mathfrak{H}$ by setting $q=e^{\pi i \tau}$ for $\tau \in \mathfrak{H}$. 

In this paper,
we concentrate on the study of the following $n$-point function related to the measure $\mathfrak{M}_q(\cdot)$:
\begin{align}\label{eqn:G as <>}
	\begin{split}
	G(t_1,t_2,\ldots,t_n)
	:=& \left\langle \prod\limits_{j=1}^{n}\left(\sum\limits_{i=1}^{\infty} t_j^{\lambda_i-i+\frac{1}{2}}\right)\right\rangle_{q}^s\\
	=&\frac{1}{\sum\limits_{\lambda\in\mathscr{P}^s}q^{|\lambda|}}
	\cdot \sum\limits_{\lambda\in\mathscr{P}^s}\prod\limits_{j=1}^{n}\sum\limits_{i=1}^{\infty} t_j^{\lambda_i-i+\frac{1}{2}}q^{|\lambda|}.
	\end{split}
\end{align}
For each given partition $\lambda\in\mathscr{P}^s$,
the series
\begin{align}\label{eqn:def prodt}
	\prod\limits_{j=1}^{n}\left(\sum\limits_{i=1}^{\infty} t_j^{\lambda_i-i+\frac{1}{2}}\right)
\end{align}
is a Laurent series in $t_j^{1/2}, j=1,2,...,n$.
Thus apparently,
the $n$-point function $G(t_1,t_2,\ldots,t_n)$ is an element in the ring
$$\mathbb{C}[\![t_j^{-1/2},t_j^{1/2}][\![q]\!].$$
With a more detailed analysis,
one can notice that the series \eqref{eqn:def prodt} is convergent provided $|t_j|>1, j=1,2,...,n$ and then it is actually a rational function in $t_j^{1/2}, j=1,2,...,n$.
Consequently,
one can regard the $n$-point function $G(t_1,t_2,\ldots,t_n)$ as a series in the ring
\begin{align*}
	\mathbb{C}(t_1^{1/2},\dots,t_n^{1/2})[\![q]\!].
\end{align*}
It is equivalent to say,
for each power $q^k$ for $k\in\mathbb{Z}_{\geq0}$,
the coefficient of $q^k$ in $G(t_1,t_2,\ldots,t_n)$ is a rational function, thus it is also a meromorphic function in the whole complex plane for every $t_j^{1/2}$.
The discussion above will be clearer after using the fermionic Fock space and the normal ordering (see subsection \ref{sec:q-diff} for more details).
This will be useful in deriving the $q$-difference equation for $n$-point function.

\subsection{Fermionic Fock space}
In this subsection, we recall the free fermions and the semi-infinite wedge construction of the fermionic Fock space.
We mainly follow the notations in \cite{DJM,O01}.

Let $S=\{s_1>s_2>\cdots\}$ be a subset of  $\mathbb{Z}+\frac{1}{2}$ consisting of half integers, and if both such subsets  
$$S_{+}:=S\backslash\{\mathbb{Z}_{\leq 0}-\frac{1}{2}\}, \qquad S_{-}:=\{\mathbb{Z}_{\leq 0}-\frac{1}{2}\}\backslash S$$  
are finite, then we say $S$ is admissible.
For a given admissible subset $S$, 
a vector associated with $S$ in the semi-infinite wedge space $\Lambda^\frac{\infty}{2}V$ is constructed by
$$v_S=\underline{s_1}\wedge\underline{s_2}\wedge\cdots.$$
The vector associated with the admissible subset $S=\mathbb{Z}_{\leq0}-\frac{1}{2}$ is called the vacuum vector $$|0\rangle=\underline{-\frac{1}{2}}\wedge\underline{-\frac{3}{2}}\wedge\underline{-\frac{5}{2}}\wedge\cdots.$$
The fermionic Fock space $\mathcal{F}=\Lambda^\frac{\infty}{2}V$ is the linear space generated by $v_S$ for all admissible $S$.
We equip the fermionic Fock space $\mathcal{F}$ with a standard inner product such that the basis $\{v_S\}$ is orthonormal.
We denote it by $(\cdot,\cdot)$.

The vacuum expectation value provides a better formalism for the inner product on the fermionic Fock space $\mathcal{F}$.
For a vector $|v\rangle\in\mathcal{F}$,
we use $\langle v|\in\mathcal{F}^*$ to denote the dual vector of $|v\rangle$,
then the vacuum expectation value is of the following form,
\begin{align*}
	\langle v| A |w\rangle
	:=\big(|v\rangle,A|w\rangle\big)
	=\big(A^*|v\rangle,|w\rangle\big),
\end{align*}
where $A$ is an operator acting on the fermionic Fock space $\mathcal{F}$ and $A^*$ is its adjoint.

The fermionic operators are two families of operators $\{\psi_{k}\}$ and $\{\psi^*_{k}\}$ labeled by half integers $k\in\mathbb{Z}+\frac{1}{2}$.
The actions of them on the fermionic Fock space $\mathcal{F}$ are defined as follows.
First,
the operator $\psi_{k}$ is the exterior multiplication by $\underline{k}$.
That is,
for any admissible $S$,
\[\psi_{k} \cdot v_S=\underline{k}\wedge v_S.\]
Then the operator $\psi^*_{k}$ is defined as the adjoint operator of $\psi_k$ with respect to the standard inner product $(\cdot,\cdot)$.
Equivalently,
\begin{equation*}
	\psi_k^*\cdot  v_S:=\left\{
	\begin{split}
		&(-1)^{l+1}\underline{s_1}\wedge \underline{s_2}\wedge\cdots\wedge\widehat{\underline{s_l}}\wedge\cdots, \quad &&\mathrm{if}\ s_l=k\ \mathrm{for\ some\ }l;  \\
		&0, \quad &&\mathrm{otherwise}.
	\end{split}\right.
\end{equation*}
One can directly verify that these two families of operators $\{\psi_{k}\}$ and $\{\psi^*_{k}\}$ satisfy the following anti-commutative relations
\begin{align}\label{eqn:anti comm}
	[\psi_{k_1},\psi_{k_2}]_+=0,
	\qquad [\psi^*_{k_1},\psi^*_{k_2}]_+=0,
	\qquad [\psi_{k_1},\psi^*_{k_2}]_+=\delta_{k_1,k_2}\cdot \text{id},
\end{align}
where the bracket is defined by $[\phi,\psi]_+=\phi\psi+\psi\phi$.
The normal ordering of product of fermions is defined as
\begin{align}\label{eqn:normal ordering}
	:\phi_{k_1}\phi_{k_2}:
	\ :=\phi_{k_1}\phi_{k_2}-\langle0|\phi_{k_1}\phi_{k_2}|0\rangle.
\end{align}
The $\phi_k$ denotes a fermion operator,
which can be either $\psi_k$ or $\psi_k^*$.
That is to say,
$:\phi_{k_1}\phi_{k_2}:$ and $\phi_{k_1}\phi_{k_2}$ only differ from each other by at most a constant.
In particular,
\begin{align}\label{eqn:normal ordering ex}
	:\psi_k \psi^*_k:
	=\psi_k \psi^*_k-\delta_{k<0}.
\end{align}

For a partition $\lambda$,
we associate an admissible subset as
\[\mathfrak{S}(\lambda):=\{\lambda_1-1/2>\lambda_2-3/2>\cdots>\lambda_i-i+1/2>\cdots\}\subset \mathbb{Z}+\frac{1}{2}.\]
We use the notation $|\lambda\rangle$ to represent the vector $v_{\mathfrak{S}(\lambda)}$, labeled by $\mathfrak{S}(\lambda)$, in the fermionic Fock space $\mathcal{F}$.
For example,
with respective to the empty partition $\emptyset$,
the associated vector is the vacuum vector $$|\emptyset\rangle=|0\rangle=\underline{-\frac{1}{2}}\wedge\underline{-\frac{3}{2}}\wedge\underline{-\frac{5}{2}}\wedge\cdots.$$
In terms of the vacuum vector and the fermionic operators $\{\psi_{k}, \psi^*_{k}\}$,
it is known that the vector $|\lambda\rangle$ can also be represented as
\begin{align}\label{eqn:lambda as fermions}
	|\lambda\rangle
	=\prod_{i=1}^{r}(-1)^{n_i}\psi_{m_i+\frac{1}{2}}\psi^*_{-n_i-\frac{1}{2}} \cdot |0\rangle,
\end{align}
where $\lambda=(m_1,...,m_r|n_1,...,n_r)$.
Actually,
all the vectors $v_S\in\mathcal{F}$ can be obtained from the action of fermions on the vacuum vector $|0\rangle$.
That is to say, for any admissible subset $S$,
$v_S$ is of the following form,
\begin{align}\label{eqn:v_S form}
	v_S=\phi_{k_1}\cdots\phi_{k_l}|0\rangle,
\end{align}
and vice versa.
Moreover,
a vector $v_{S}$ comes from a partition,
i.e. $v_S=|\lambda\rangle$ for some partition $\lambda$,
if and only if
\begin{align*}
	|S_{+}|=|S_{-}|.
\end{align*}
The subspace of $\mathcal{F}$ generated by $|\lambda\rangle$ is particularly interesting,
and we call it the charge zero fermionic Fock space $\mathcal{F}_0$.
From the definition,
it is apparent that $\mathcal{F}_0$ has a orthonormal basis labeled by all partitions.
Thus,
this space $\mathcal{F}_0$ is widely used in studying properties of all partitions and especially generating functions weighted by functions related partitions (see, for examples, \cite{BO,Dij95,HIL,TW07,Y23,Y24}).
This is indeed the method used in \cite{BO,O01} to study the correlation function of all integer partitions.
In Section \ref{sec:main result},
we will introduce the notion of $\omega$-transform on fermions and fermionic Fock space.
From that,
one can use $\mathcal{F}_0$ and $\omega$-transform to directly study self-conjugate partitions.

\subsection{Charge, energy and translation operators}
\label{sec:charge}
We review three canonical operators commonly used in the language of fermionic Fock space.

The charge operator $C$ and the energy operator $H$ are defined by
\[C=\sum_k:\psi_k\psi^*_k:
\qquad\text{and}
\qquad H=\sum_kk:\psi_k\psi^*_k:,\]
respectively.
It is direct to verify
\begin{align}\label{eqn:action of psipsi}
	\psi_k \psi^*_k \cdot v_S
	=\begin{cases}
		v_S,\qquad &k\in S,\\
		0, \qquad &k\notin S
	\end{cases}
\end{align}
from the definition of $v_S$.
Then,
from equation \eqref{eqn:normal ordering ex},
the actions of $C$ and $H$ are given by
\begin{align}\label{eqn:C action}
	C \cdot v_S=\big(|S_+|-|S_-|\big)v_S,
\end{align}
and
\begin{align}\label{eqn:H action}
	H \cdot v_S= \Big(\sum_{s\in S_+} s- \sum_{s\in S_-} s\Big) v_S.
\end{align}

For a vector $|v\rangle\in\mathcal{F}$,
if it is an eigenvector of $C$,
we say that $|v\rangle$ is of pure charge and its charge is exactly defined by the corresponding eigenvalue.
Similarly,
the eigenvalue of $|v\rangle$ with respective to $H$ is called its energy.
From equations \eqref{eqn:C action} and \eqref{eqn:H action},
$v_S$ is of pure charge and energy for each admissible subset $S$.
Especially,
for any partition $\lambda$,
\[C\cdot |\lambda\rangle=0,
\qquad\text{and}\qquad H \cdot |\lambda\rangle=|\lambda|\cdot |\lambda\rangle.\]

The translation operator $R$ is defined by
\[R \cdot \underline{s_1}\wedge\underline{s_2}\wedge\underline{s_3}\cdots:=\underline{s_1+1}\wedge\underline{s_2+1}\wedge\underline{s_3+1}\wedge\cdots\]
for any admissible subset $S=\{s_1>s_2>s_3>...\}$,
and then the inverse of $R$ is
\[R^{-1}
\cdot \underline{s_1}\wedge\underline{s_2}\wedge\underline{s_3}\cdots=\underline{s_1-1}\wedge\underline{s_2-1}\wedge\underline{s_3-1}\wedge\cdots.\]
It follows that the commutation relations of $R$ and the operators $\psi_k, \psi_k^*, C, H$ are given by
\begin{align}
	R^{-k}\psi_iR^{k}=&\psi_{i-k},
	\qquad\ R^{-k}\psi_i^*R^{k}=\psi_{i-k}^*, \label{eqn:PphiR}\\
	R^{-k} C R^{k}=&C+k,
	\qquad R^{-k}HR^k=H+kC+\frac{k^2}{2}.
\end{align}

The charge endows the fermionic Fock space $\mathcal{F}$ with a structure of graded space. In another words, we have the following charge decomposition
\begin{align*}
	\mathcal{F}=\bigoplus_{k\in \Z} R^k\mathcal{F}_0
\end{align*}
where $\mathcal{F}_0=\mathrm{ker} C$ is exactly the charge zero subspace introduced at the end of last subsection. We say that an operator is of charge $k$ if it sends an element of charge $i$  to an element of charge $i+k$. It is obvious that $\psi_k$, $\psi^*_k$ and $R$ have charges $+1,-1$ and $+1$, respectively.

\section{The $n$-point function and $q$-difference equations}
\label{sec:main result}
In this section,
we introduce the $\omega$-transform on the fermionic operators and the fermionic Fock space.
By virtue of that,
we derive the $q$-difference equation for the $n$-point function $G(t_1,t_2,...,t_n)$ related to the measure $\mathfrak{M}_q(\cdot)$.

\subsection{The $\omega$-transform}
	
We introduce the $\omega$-transform on fermionic operators and extend it to the whole fermionic Fock space $\mathcal{F}$ through fermionic action.
It is motivated by the method used in subsection 2.4 of the second-named author's paper \cite{Y23} and the involution,
which maps the elementary symmetric functions to the complete symmetric functions,
on the ring of symmetric functions (see, for example, Chapter I.2 in \cite{Mac}).
In this subsection,
we shall use $\phi_k$ to denote a fermion of either type, i.e., $\phi_k=\psi_k$ or $\phi_k=\psi_k^*$.  
	
\begin{defn}[$\omega$-transform on fermionic operators] For a single $\psi_k$ or $\psi_k^*$, 
	\[\omega(\psi_k):=(-1)^{k+\frac{1}{2}}\psi_{-k}^*,
	\qquad \omega(\psi_k^*):=(-1)^{k+\frac{1}{2}}\psi_{-k}.\]
	For a product of fermions $\phi_{k_1},\ldots,\phi_{k_j}$, 
	$$\omega(\phi_{k_1}\cdots\phi_{k_j}):=\omega(\phi_{k_1})\cdots\omega(\phi_{k_j}).$$
	We extend this $\omega$-transform to the linear space spanned by products of fermions.
	In particular, $\omega(0):=0$ and $\omega(\rm{id}):=\rm{id}$.
\end{defn}

From the definition of the $\omega$-transform,
we have
\begin{align}\label{eqn:w *}
	\omega(\phi_k^*)=\omega(\phi_k)^*.
\end{align}
Moreover,
the $\omega$-transform on the fermionic operators has the following properties.
\begin{lem}\label{lem:basic of w}
	We have
	\hspace*{\fill}
	\begin{enumerate}
		\item $\omega(\omega(\phi_k))=-\phi_k.$
		\item $\omega([\phi_{k_1},\phi_{k_2}]_+)=[\omega(\phi_{k_1}),\omega(\phi_{k_2})]_+.$
	\end{enumerate}
	\end{lem}
\begin{proof}
	By definition, $$\omega(\omega(\phi_k))=\omega((-1)^{k+\frac{1}{2}}\phi_{-k}^*)=(-1)^{k+\frac{1}{2}}\cdot(-1)^{-k+\frac{1}{2}}\phi_k=-\phi_k.$$
	On the other hand,
	\begin{align*}
		\omega([\phi_{k_1},\phi_{k_2}]_+)&=\omega(\phi_{k_1}\phi_{k_2}+\phi_{k_2}\phi_{k_1})\\
		&=\omega(\phi_{k_1})\omega(\phi_{k_2})+\omega(\phi_{k_2})\omega(\phi_{k_1})\\
		&=[\omega(\phi_{k_1}),\omega(\phi_{k_2})]_+.
	\end{align*}
\end{proof}

\begin{rmk}
	If we view the linear space generated by products of fermions as a super Lie algebra, then the lemma above says that the $\omega$-transform on fermionic operators is a super Lie algebra endomorphism.	
\end{rmk}

Further,
we still need to verify the algebraic compatibility of the $\omega$-transform since the fermions $\{\psi_{k},\psi^*_{k}\}$ are not freely generated.
Actually,
by Lemma \ref{lem:basic of w},
we have
\begin{align*}
	&\omega(\delta_{k_1,k_2} \cdot \text{id})
	=\omega([\psi_{k_1},\psi^*_{k_2}]_+)
	=[\omega(\psi_{k_1}),\omega(\psi^*_{k_2})]_+\\
	&\qquad\qquad\qquad=(-1)^{k_1+k_2+1}[\psi^*_{-k_1},\psi_{-k_2}]_+
	=\delta_{k_1,k_2} \cdot \text{id},\\
	&\omega(0)
	=\omega([\psi_{k_1},\psi_{k_2}]_+)
	=[\omega(\psi_{k_1}),\omega(\psi_{k_2})]_+
	=0,
\end{align*}
and similarly, $\omega([\psi^*_{k_1},\psi^*_{k_2}]_+)=0$,
where we have used that $k_1,k_2$ are half-integers.
Thus,
the $\omega$-transform is compatible with the anti-commutation relations \eqref{eqn:anti comm} for fermions.

\begin{defn}[$\omega$-transform on fermionic Fock space]
	For a vector $v=\phi_{k_1}\cdots\phi_{k_j} |0\rangle\in\mathcal{F}$,
	define
	\[\omega(v):=\omega(\phi_{k_1})\cdots\omega(\phi_{k_j})
	|0\rangle\in\mathcal{F}.\]
	Similarly,
	for a vector $v^*=\langle0|\phi_{k_1}\cdots\phi_{k_j} \in\mathcal{F}^*$,
	define
	\[\omega(v^*):=\langle0|\omega(\phi_{k_1})\cdots\omega(\phi_{k_j})\in\mathcal{F}^*.\]
	From equation \eqref{eqn:w *},
	one has
	\begin{align}\label{eqn:w *F}
		\omega(v^*)=\omega(v)^*.
	\end{align}
\end{defn}

When restricting to the charge zero fermionic Fock space $\mathcal{F}_0$,
the $\omega$-transform has the following interesting properties that we are going to apply later.
\begin{lem}\label{lem:w la la^t}
	When restricting to the charge zero fermionic Fock space $\mathcal{F}_0$,
	we have
	\begin{align}
		\omega(|\lambda\rangle)=|\lambda^t\rangle.
	\end{align}
	Similarly,
	when considering the dual charge zero space $\mathcal{F}_0^*$,
	we have
	\begin{align}
		\omega(\langle\lambda|)=\langle\lambda^t|.
	\end{align}
\end{lem}
\begin{proof}
	We suppose
	$\lambda=(m_1,...,m_r|n_1,...,n_r)$.
	Then from the fermionic representation \eqref{eqn:lambda as fermions} of $|\lambda\rangle$,
	we have
	\begin{align*}
		\omega (|\lambda\rangle)&=\omega\left(\prod_{i=1}^{r}(-1)^{n_i}\psi_{m_i+\frac{1}{2}}\psi^*_{-n_i-\frac{1}{2}}
		\cdot |0\rangle\right) \\
		&=\prod_{i=1}^{r}(-1)^{n_i+m_i+1-n_i}\psi^*_{-m_{i}-\frac{1}{2}}\psi_{n_i+\frac{1}{2}}
		\cdot|0\rangle\\
		&=\prod_{i=1}^{r}(-1)^{m_i}\psi_{n_i+\frac{1}{2}}\psi^*_{-m_{i}-\frac{1}{2}}
		\cdot|0\rangle=|\lambda^t\rangle.
	\end{align*}
	By taking the dual of the equation above,
	and the property of $\omega$-transform in equation \eqref{eqn:w *F},
	we have $\omega(\langle\lambda|)=\langle\lambda^t|$.
\end{proof}

\begin{lem}\label{omega}
	The $\omega$-transform keeps the inner product.
	That is to say,
	\begin{align}\label{eqn:w ()}
		\big(\omega(v_{S_1}),\omega(v_{S_2})\big)=(v_{S_1},v_{S_2}).
	\end{align}
	In particular,
	for any two partitions $\lambda,\mu$, and a family of fermions $\phi_{k_j}$,
	we have
	\begin{align}\label{eqn:<>=<w>}
		\langle\mu|\phi_{k_1}\cdots\phi_{k_j}|\lambda^t\rangle=\langle\mu^t|\omega(\phi_{k_1})\cdots\omega(\phi_{k_j})|\lambda\rangle.
	\end{align}
\end{lem}

\begin{proof}
	For the first equation \eqref{eqn:w ()},
	we first recall that any vector $v_{S}$ is of the form $\phi_{k_1}\cdots\phi_{k_1}|0\rangle$,
	so is $w(v_S)$.
	Thus,
	from the definition that $\{v_S\}_{S\ \text{is admissible}}$ forms a orthonormal basis of the fermionic Fock space $\mathcal{F}$,
	we have
	\begin{align*}
		(v_{S_1},v_{S_2})=1 &\Leftrightarrow v_{S_1}=v_{S_2}\\
		&\Leftrightarrow\omega(v_{S_1})=\omega(v_{S_2})\\
		&\Leftrightarrow\big(\omega(v_{S_1}),\omega(v_{S_2})\big)=1.
	\end{align*}
	The second equation \eqref{eqn:<>=<w>} directly follows from equation \eqref{eqn:w ()} and Lemma \ref{lem:w la la^t}.
\end{proof}

\begin{rmk}
	The equation \eqref{eqn:w ()} is a generalization of the equation (2.10) in the second-named author's paper \cite{Y23}.
\end{rmk}

Since each $\phi_k= \psi_k$ or $\psi^*_k$ is of charge $\pm1$ and the $\omega$-transform does not keep the charge,
we sometimes need to deal with the charge $\pm 1$ Fock spaces but not only $\mathcal{F}_0$ when deriving the $q$-difference equation for the $n$-point functions $G(t_1,t_2,...,t_n)$ in the next subsection.
The following lemma will also be useful.
\begin{lem}\label{omega1}
	When restricting to the charge $\pm 1$ fermionic Fock space,
	we have
	\begin{align}\label{eqn:w c1}
		\omega(R|\lambda\rangle)=-R^{-1}|\lambda^t\rangle,
	\end{align}
	and
	\begin{align}\label{eqn:w c-1}
		\qquad\omega(R^{-1}|\lambda\rangle)=R|\lambda^t\rangle.
	\end{align}
\end{lem}
\begin{proof}
	From equation \eqref{eqn:lambda as fermions} and the commutation relation \eqref{eqn:PphiR},
	we first have
	\begin{align*}
		\omega(R|\lambda\rangle)&=\omega\Big(R\prod\limits_{i=1}^{r}(-1)^{n_i}\psi_{m_i+\frac{1}{2}}\psi_{-n_i-\frac{1}{2}}^*
		\cdot|0\rangle\Big)\\
		&=\omega\Big(\prod\limits_{i=1}^{r}(-1)^{n_i}\psi_{m_i+\frac{3}{2}}\psi_{-n_i+\frac{1}{2}}^*\cdot
		R |0\rangle\Big).
	\end{align*}
	Then,
	from $R|0\rangle=\psi_{\frac{1}{2}}|0\rangle, R^{-1}|0\rangle=\psi_{-\frac{1}{2}}^*|0\rangle$,
	and the definition of $\omega$-transform on the fermionic operators,
	the formula above is equal to
	\begin{align*}
		\begin{split}
		&\prod\limits_{i=1}^{r}(-1)^{n_i+m_i+2-n_i+1}\psi_{-m_i-\frac{3}{2}}^*\psi_{n_i-\frac{1}{2}}\cdot(-1)\cdot R^{-1}
		|0\rangle\\
		&
		\qquad\qquad=-R^{-1}\prod\limits_{i=1}^{r}(-1)^{m_i}\psi_{n_i+\frac{1}{2}}\psi_{-m_i-\frac{1}{2}}^*
		\cdot |0\rangle=-R^{-1}|\lambda^t\rangle,
		\end{split}
	\end{align*}
	which proves equation \eqref{eqn:w c1}.
	The second equation \eqref{eqn:w c-1} can be proved similarly.
\end{proof}

\subsection{The $q$-difference equation for the $n$-point function}
\label{sec:q-diff}
In this subsection,
we apply $\omega$-transform to deduce the $q$-difference equation for the $n$-point function $G(t_1,\cdots,t_n)$.

Recall that the $n$-point function $G(t_1,\cdots,t_n)$ is defined by
\begin{align*}
	G(t_1,t_2,\ldots,t_n)
	=\frac{1}{\sum\limits_{\lambda\in\mathscr{P}^s}q^{|\lambda|}}
	\cdot\sum\limits_{\lambda\in\mathscr{P}^s}\prod\limits_{j=1}^{n}\sum\limits_{i=1}^{\infty} t_j^{\lambda_i-i+\frac{1}{2}}q^{|\lambda|}.
\end{align*}
It is obvious that $G(t_1,t_2,\ldots,t_n)$ is symmetric with respective to all variables $t_j, 1\leq j\leq n$.
Define the function $f_n:\mathscr{P}\rightarrow \mathbb{C}[\![t_j^{-1/2},t_j^{1/2}]$ as
\begin{align*}
	f_n(\lambda):=
	\prod\limits_{j=1}^{n}\sum\limits_{i=1}^{\infty} t_j^{\lambda_i-i+\frac{1}{2}}.
\end{align*}
The $n$-point function $G(t_1,t_2,...,t_n)$ can be represented as the $q$-bracket of the function $f_n$ as in equation \eqref{eqn:G as <>}.
From the discussion at the end of subsection \ref{sec:def m},
the images of $f_n(\cdot)$ could be regarded as elements in the ring $\mathbb{C}(t_1^{1/2},\dots,t_n^{1/2})$.
Following Okounkov \cite{O01},
we introduce the following auxiliary operator
\[T(t):=\sum\limits_{k \in \Z+\frac{1}{2}}t^k\psi_k\psi_k^*.\]
From equation \eqref{eqn:action of psipsi},
we have
\begin{align*}
	\prod_{j=1}^nT(t_j)\cdot |\lambda\rangle
	=f_n(\lambda) \cdot |\lambda\rangle.
\end{align*}
Thus,
the $n$-point function $G(t_1,t_2,...,t_n)$ can be represented in terms of the vacuum expectation values as the following form
\begin{align}\label{eqn:G as vev}
	G(t_1,t_2,\ldots,t_n)
	=\frac{1}{\sum\limits_{\lambda\in\mathscr{P}^s} \langle\lambda|q^{H}|\lambda^t\rangle}
	\cdot\sum\limits_{\lambda\in\mathscr{P}^s}\langle\lambda| q^{H}\prod\limits_{j=1}^{n}T(t_j) |\lambda^t\rangle.	
\end{align}
Moreover,
since for any partition $\lambda$,
\begin{align*}
	\langle\lambda| q^{H} |\lambda^t\rangle
	=q^{|\lambda^t|} \cdot \langle\lambda|\lambda^t\rangle
\end{align*}
and
\begin{align*}
	\langle\lambda| q^{H}\prod\limits_{j=1}^{n}T(t_j) |\lambda^t\rangle
	=q^{|\lambda^t|}f_n(\lambda^t) \cdot \langle\lambda|\lambda^t\rangle
\end{align*}
vanish unless $\lambda=\lambda^t$, i.e. $\lambda$ is self-conjugate,
we can extend the summation $\sum_{\lambda\in\mathscr{P}^s}$ in equation \eqref{eqn:G as vev} to the summation over the set of all integer partitions.

Let the normal ordering of $T(t)$ be 
\[:T(t):=\sum\limits_{k \in \Z+\frac{1}{2}}t^k:\psi_k\psi_k^*:.\]
Then from equation \eqref{eqn:normal ordering ex},
we have
\begin{align}\label{normalordering}
	T(t)=:T(t):+\sum\limits_{k=-\infty}^{1/2}t^k=:T(t):+\frac{1}{t^{1/2}-t^{-1/2}},
\end{align}
where we have used $|t|>1$ in the second equality.
Actually,
for any given partition $\lambda$,
the vector $|\lambda\rangle\in\mathcal{F}$ is an eigenvector of $:T(t):$,
whose eigenvalue is a polynomial in $t^{1/2}$ and $t^{-1/2}$.
Thus,
the action of $T(t)$ on a given partition could be regarded as multiplying a rational function in $t^{\half}$,
which is exactly compatible with our explanation at the end of subsection \ref{sec:def m}.
As a consequence,
under such consideration,
which regards the $n$-point function $G(t_1,t_2,\ldots,t_n)$ as a series in the ring $\mathbb{C}(t_1^{1/2},\dots,t_n^{1/2})[\![q]\!]$,
the equation \eqref{normalordering} holds for all $t\in\mathbb{C}$ and could be regarded as the meromorphic continuation of $T(t)$.

\begin{thm}[=Theorem \ref{main}] The $n$-point function satisfies the following $q$-difference equation
	\begin{equation}\label{eqn:q-diff}
		\begin{aligned}
			&G(q^{-2}t_1,t_2,\ldots,t_n)\\
			&\qquad\quad=\sum\limits_{P^-,P^+ \subseteq \{2\ldots,n\},\atop P^-\cap P^+=\emptyset}(-1)^{|P^-|-1}G(t_1\prod\limits_{i\in P^-}t_i^{-1}\cdot \prod_{j\in P^+}t_j,\cdots,\hat{t_i},\cdots,\hat{t_j},\cdots),
		\end{aligned}
	\end{equation}
	where the notation $\hat{\cdot}$ means that the corresponding term should be omitted.
	Parallel formula for $G(t_1,\ldots,q^{-2}t_j,\ldots,t_n)$ can be achieved by exchanging the variables $t_1,...,t_n$.
\end{thm}
\begin{proof}
	The strategy of our proof is to use the expression \eqref{eqn:G as vev} of the $n$-point function $G(t_1,t_2,\ldots,t_n)$.
	We first split the operator $T(t_1)$ in the right hand side of equation \eqref{eqn:G as vev} and then recover it by applying the $\omega$-transform twice. 
	This will produce the $q$-difference equation \eqref{eqn:q-diff} as desired.
	
	Let $Z_s(q)=\sum_{\lambda\in\mathscr{P}^s} \langle\lambda|q^{H}|\lambda^t\rangle$ be the normalization factor. Recall that in the right hand side of equation \eqref{eqn:G as vev}, since all $T(t_i)$ commute with each other, we can move $T(t_1)$ to the last place for convenience and then split it there.
	First,
	\begin{align}\label{eqn:ZsG}
		\begin{split}
		Z_s(q) \cdot G(t_1,t_2,\ldots,t_n)
		&=\sum\limits_{\lambda\in\mathscr{P}}\langle \lambda | q^{H}\prod\limits_{j=1}^{n}T(t_j)| \lambda^t \rangle\\
		&=\sum\limits_{k\in\Z+\frac{1}{2}}t_1^k\sum\limits_{\lambda\in\mathscr{P}}\langle \lambda | q^{H}\prod\limits_{j=2}^{n}T(t_j)\cdot \psi_k\psi^*_k| \lambda^t \rangle.     
		\end{split}
	\end{align}
	Then, 
	the procedure of splitting $T(t_1)$ is to insert the operator
	\begin{align}\label{eqn:Nmu}
		\sum_{N\in \Z}\sum_{\mu\in\mathscr{P}}R^N|\mu\rangle\cdot \langle\mu|R^{-N}
	\end{align}
	to the middle of fermions $\psi_k$ and $\psi^*_k$ in the second line of equation \eqref{eqn:ZsG},
	since the operator \eqref{eqn:Nmu} is the identity operator on the fermionic Fock space $\mathcal{F}$. 
	Remark that we cannot only use the $N=0$ part of the operator \eqref{eqn:Nmu} since a single fermion $\psi_k$ or $\psi_k^*$ is not an operator on the charge zero fermionic Fock space $\mathcal{F}_0$, but on the total space $\mathcal{F}$.
	After doing that, only $N=-1$ part of the operator \eqref{eqn:Nmu} survives since $R^{-N}\psi^*_k$ is of charge $-N-1$
	(see subsection \ref{sec:charge}).
	The result is, 
	\begin{align}\label{eqn:mG def}
		\begin{split}
			Z_s(q) \cdot G&(t_1,t_2,\ldots,t_n)\\
			&=\sum\limits_{k\in\Z+\frac{1}{2}}t_1^k\sum\limits_{\lambda,\mu\in\mathscr{P}}\langle \lambda | q^{H}\prod\limits_{j=2}^{n}T(t_j)\cdot \psi_kR^{-1}|\mu\rangle\cdot\langle\mu|R\psi_k^*|\lambda^t\rangle.
		\end{split}
	\end{align}
	For the last part of the equation above,
	we apply the $\omega$-transform and Lemmas \ref{omega}, \ref{omega1} to $\langle\mu|R\psi_k^*|\lambda^t\rangle$.
	The result is
	\begin{align*}
		\langle\mu|R\psi_k^*|\lambda^t\rangle&=\langle\mu|\psi_{k+1}^*R|\lambda^t\rangle=\langle\mu^t|\psi_{-k-1}R^{-1}|\lambda\rangle\cdot(-1)^{k+\frac{1}{2}}\\
		&=\langle\mu^t|R^{-1}\psi_{-k}|\lambda\rangle\cdot(-1)^{k+\frac{1}{2}}.
	\end{align*}
	Then by substituting the result above to equation \eqref{eqn:mG def} and taking summation over $\lambda$,
	\begin{align}\label{ZsA}
		Z_s(q) \cdot G(t_1,t_2,\ldots,t_n)=\sum\limits_{k\in\Z+\frac{1}{2}}t_1^k\cdot \sum\limits_{\mu\in\mathscr{P}}(-1)^{k+\frac{1}{2}}A_\mu,
	\end{align}
	where
	\begin{align}\label{eqn:def Amu}
		A_\mu:=\langle \mu^t | R^{-1}\psi_{-k}q^{H}\prod\limits_{j=2}^{n}T(t_j)
		\cdot\psi_kR^{-1}|\mu\rangle.
	\end{align}
	In consideration of our target, we should move $\psi_{-k}$ in equation \eqref{eqn:def Amu} to the right hand side of $\psi_{k}$, then repeat the splitting procedure and $\omega$-transform again.
	After applying commutation relations
	\begin{align}
		[T(t),\psi_k^*]=-t^k\psi_k^*,
		\qquad\text{and}\qquad[T(t),\psi_k]=t^k\psi_k,
		\label{Tpsicomm} 
	\end{align} 
	which can be deduced from equation \eqref{eqn:anti comm}, we have
	\begin{align*}
		A_\mu=-q^{k}\langle \mu^t | R^{-1}q^{H}\cdot\bigg(\sum\limits_{P \subseteq \{2,\ldots,n\}}(-1)^{|P|}
		\prod\limits_{i\in P}t_i^{-k}
		\cdot\prod_{i \notin P}T(t_i)\bigg)
		\psi_k\psi_{-k}R^{-1}|\mu\rangle.
	\end{align*}
	For convenience, let 
	$$\mathbb{T}=\sum\limits_{P \subseteq \{2,\ldots,n\}}(-1)^{|P|}
	\prod\limits_{i\in P}t_i^{-k}
	\cdot\prod_{i \notin P}T(t_i),$$
	then we run the splitting procedure and $\omega$-transform again to obtain (here we omit the computation details and only write down the results)
	\begin{align*}
		A_\mu&=-q^{k}\langle \mu^t | R^{-1}q^{H}\cdot\mathbb{T}
		\psi_k\psi_{-k}R^{-1}|\mu\rangle\\
		&=-\sum\limits_{\lambda\in\mathscr{P}}q^{k}\langle \mu^t | R^{-1}q^{H}\cdot\mathbb{T}
		\psi_k|\lambda\rangle\langle\lambda|\psi_{-k}R^{-1}|\mu\rangle\\
		&=-\sum\limits_{\lambda\in\mathscr{P}}q^{k}\langle \mu^t | R^{-1}q^{H}\cdot\mathbb{T}
		\psi_k|\lambda\rangle\langle\lambda^t|\psi_k^*R|\mu^t\rangle\cdot(-1)^{-k+\frac{1}{2}}.
	\end{align*}
	Inserting the equation above back to the equation \eqref{ZsA} and taking summation over $\mu^t$, we have 
	\begin{align*}
		\begin{split}
		Z_s(q) \cdot G(t_1,t_2,\ldots,t_n)
		=\sum\limits_{k\in\Z+\frac{1}{2}}t_1^k&\sum\limits_{P \subseteq \{2,\ldots,n\}}(-1)^{|P|}
		\prod\limits_{i\in P}t_i^{-k}\\
		&\cdot\sum\limits_{\lambda\in\mathscr{P}}\langle \lambda^t|\psi_{k}^*q^{k}q^{H}
		\prod_{i \notin P}T(t_i) \cdot \psi_k|\lambda\rangle.
		\end{split}
	\end{align*}
	Again by commutation relations \eqref{Tpsicomm},
	the equation above is reduced to 
	\begin{align}\label{eqn:mG as B}
		Z_s(q) \cdot G(t_1,t_2,\ldots,t_n)
		=\sum\limits_{P \subseteq \{2\ldots,n\}}(-1)^{|P|}B_\lambda,
	\end{align}
	where
	$$B_\lambda=\sum\limits_{\lambda\in\mathscr{P}}\langle \lambda^t|q^{H}\bigg(\sum\limits_{P' \subseteq P^c}\tilde{T}\Big(q^2t_1
	\prod\limits_{i\in P}t_i^{-1}
	\cdot\prod_{j\in P'}t_j\Big)\prod_{j \in P^c\backslash P'}T(t_j)\bigg)|\lambda\rangle.$$
	Here we have used an another auxiliary operator $\tilde{T}(t)$ defined by
	$$\tilde{T}(t):=\sum\limits_{k \in \Z+\frac{1}{2}}t^k\psi_k^*\psi_k.$$
	
	Note that from the definition of normal ordering and equation \eqref{eqn:normal ordering ex},
	there are relations
	\begin{align}
		T(t)=&:T(t):+\frac{1}{t^{1/2}-t^{-1/2}},\qquad|t|>1, \label{eqn:T and :T:}\\
		\tilde{T}(t)=&-:T(t):-\frac{1}{t^{1/2}-t^{-1/2}},\qquad|t|<1
	\end{align}
	as Laurent series in $\mathbb{C}[\![t^{-1/2},t^{1/2}]$.
	Then, in the sense of meromorphic continuations of $T(t)$ and $\tilde{T}(t)$ (see discussion at the end of subsection \ref{sec:def m}),
	i.e., regarding the dependence of them on $t^{1/2}$ as meromorphic functions in $\mathbb{C}$,
	we have
	$\tilde{T}(t)=-T(t).$
	Consequently,
	equation \eqref{eqn:mG as B} gives
	\begin{align}\label{eqn:mG last}
		Z_s(q) \cdot G(t_1,t_2,\ldots,t_n)=\sum\limits_{P \subseteq \{2\ldots,n\}}(-1)^{|P|-1}C_\lambda,
	\end{align}
	where \[C_\lambda:=\sum\limits_{\lambda\in\mathscr{P}}\langle \lambda^t|q^{H}\bigg(\sum\limits_{P' \subseteq P^c}T\Big(q^2t_1\prod\limits_{i\in P}t_i^{-1}
	\cdot \prod_{j\in P'}t_j\Big)
	\cdot \prod_{j \in P^c\backslash P'}T(t_j)\bigg)|\lambda\rangle.\]
	Dividing both sides of the equation \eqref{eqn:mG last} by $Z_s(q)$ and reorganizing the indices, we then obtain the following $q$-difference equation for the $n$-point function
	\begin{align*}
		G(t_1&,t_2,\ldots,t_n)\\
		&=\sum\limits_{P^-,P^+ \subseteq \{2\ldots,n\}, \atop P^-\cap P^+=\emptyset}(-1)^{|P^-|-1}G(q^2t_1
		\prod\limits_{i\in P^-}t_i^{-1}
		\cdot \prod_{j\in P^+}t_j,\cdots,\hat{t_i},\cdots,\hat{t_j},\cdots),
	\end{align*}
	which is equivalent to the equation \eqref{eqn:q-diff} presented in the statement of this theorem.
\end{proof}

\section{Applications of the $q$-difference equation} 
\label{sec:app}
In this section, we derive closed formulas of the one-point function $G(t)$ and the two-point function $G(t_1,t_2)$ using Theorem \ref{main}.
These explicit formulas only involve theta functions $\Theta_1(t;q), \Theta_3(t;q)$, and then inherit the quasimodularity of these functions.

From now on,
we always assume $0<|q|<1$. 

\subsection{An explicit formula for the one-point function}
In this subsection,
we derive the explicit formula for the one-point function presented in Corollary \ref{cor:main one}.
According to Theorem \ref{main}, the one-point function 
\begin{align*}
	G(t)&=\left\langle \sum\limits_{i=1}^{\infty} t^{\lambda_i-i+\frac{1}{2}} \right\rangle^s_{q}
\end{align*}
satisfies the following $q$-difference equation
\begin{align}\label{oneqdiff}
	G(q^2t)=-G(t).
\end{align}
To obtain the explicit formula for the one-point function $G(t)$,
we need to analyze the singularity of $G(t)$ and solve the $q$-difference equation \eqref{oneqdiff}.

We first review some known properties of the theta functions $\Theta_1(t;q)$ and $\Theta_3(t;q)$. 
They are defined by
\begin{align*}
	\Theta_1(t;q)&:=\sum\limits_{n\in \Z}(-1)^nq^{(n+\frac{1}{2})^2}t^{n+\frac{1}{2}},
\end{align*}
and
\begin{align*}
	\Theta_3(t;q)&:=\sum\limits_{n\in \Z}q^{n^2}t^{n}.
\end{align*}
Note that we have
$\Theta_1(t;q)
=t^{\frac{1}{2}}q^{\frac{1}{4}}\Theta_3(-tq;q)$.
We use the following notations for derivatives
\begin{align*}
	\Theta_1^{'}(t;q):=t\frac{\partial}{\partial t}\Theta_1(t;q), \qquad\Theta_3^{'}(t;q):=t\frac{\partial}{\partial t}\Theta_3(t;q). 
\end{align*}

\begin{lem}\label{lem:theta q-diff}
	The theta functions $\Theta_1(t;q)$ and $\Theta_3(t;q)$ satisfy the following $q$-difference equations,
	\begin{align}
		\Theta_1(q^{-2}t;q)&=-q^{-1}t \cdot \Theta_1(t;q),\\
		\Theta_3(q^{-2}t;q)&=q^{-1}t \cdot \Theta_3(t;q).
	\end{align}
\end{lem}
\begin{proof}
	This directly follows from the definitions of theta functions.
\end{proof}

\begin{cor}\label{cor:theta' q-diff}
	We have
	\begin{align}
		\Theta_1^{'}(q^{-2}t;q)&=-q^{-1}t \cdot \big(\Theta_1(t;q)+\Theta_1^{'}(t;q)\big),\\
		\Theta_3^{'}(q^{-2}t;q)&=q^{-1}t \cdot \big(\Theta_3(t;q)+\Theta_3^{'}(t;q)\big).
	\end{align}
\end{cor}

\begin{lem}\label{lem:theta inf prod}
	The theta functions $\Theta_1(t;q)$ and $\Theta_3(t;q)$ have the following infinite product decomposition,
	\begin{align}
		\Theta_1(t;q)&=q^{1/4}(t^{1/2}-t^{-1/2})
		\cdot \prod_{m=1}^{\infty}(1-q^{2m})(1-q^{2m}t)(1-q^{2m}/t),\\
		\Theta_3(t;q)&=\prod_{m=1}^{\infty}(1-q^{2m})(1+q^{2m-1}t)(1+q^{2m-1}/t).
	\end{align}
	As a result,
	both of them could be regarded as meromorphic functions for the variable $t$ in the whole complex plane $\mathbb{C}$. Moreover, 
	\begin{align}
		\Theta_1(t^{-1};q)=-\Theta_1(t;q), \qquad\Theta_3(t^{-1};q)=\Theta_3(t;q).
	\end{align}
\end{lem}
\begin{proof}
	This is exactly the well-known triple product formula.
\end{proof}

\begin{prop}\label{prop:1-point}
	The one-point function $G(t)$ admits a meromorphic continuation to $t\in\mathbb{C}$
	and more explicitly,
	\begin{align}\label{eqn:one point formula}
		G(t)&=q^{1/4}\prod_{m=1}^{\infty}\frac{(1-q^{2m})^2}{(1+q^{2m-1})^2}\cdot\frac{\Theta_3(t;q)}{\Theta_1(t;q)}.
	\end{align}
\end{prop}
\begin{proof}
First,
we consider the following one-point function with normal ordering
\begin{align*}
	:G(t):
	:=\frac{1}{\sum\limits_{\lambda\in\mathscr{P}^s} \langle\lambda|q^{H}|\lambda^t\rangle}
	\cdot\sum\limits_{\lambda\in\mathscr{P}^s}\langle\lambda| q^{H}:T(t): |\lambda^t\rangle.
\end{align*}
We shall show that $:G(t):$ is absolutely convergent in the following region
\begin{align*}
	\Delta_{\epsilon,1}:=\big\{t\in \mathbb{C}\big|
	\ |q|^{2-\epsilon}<|t|<|q|^{-2+\epsilon}\big\},
\end{align*}
thus it can be regarded as a holomorphic function.

In fact, set $\sigma_\lambda(t)=\sum_{k\in\mathfrak{S}_+(\lambda)}t^k-\sum_{k\in\mathfrak{S}_-(\lambda)}t^k$,
then from equation \eqref{eqn:action of psipsi},
we have
\[:T(t):\cdot |\lambda\rangle
=\sigma_\lambda(t) \cdot |\lambda\rangle\]
for any vector $|\lambda\rangle\in\mathcal{F}_0$.
We need to estimate $\sigma_\lambda(t)$ and this will split into two cases.
Let $\epsilon$ be a small positive number, then
\begin{align}\label{eqn:t>1 sigma}
	|\sigma_\lambda(t)|
	\leq|t|^{||\mathfrak{S}_+(\lambda)||}+|\mathfrak{S}_-(\lambda)|
	\leq|q|^{(-2+\epsilon)||\mathfrak{S}_+(\lambda)||}+|\mathfrak{S}_-(\lambda)|
\end{align}
for $1\leq|t|< |q|^{-2+\epsilon}$,
and
\begin{align}\label{eqn:t<1 sigma}
	|\sigma_\lambda(t)|
	\leq|t|^{-||\mathfrak{S}_-(\lambda)||}+|\mathfrak{S}_+(\lambda)|
	\leq|q|^{(-2+\epsilon)||\mathfrak{S}_-(\lambda)||}+|\mathfrak{S}_+(\lambda)|
\end{align}
for $|q|^{2-\epsilon}<|t|<1$,
where  $$||\mathfrak{S}_+(\lambda)||:=\sum\limits_{k\in\mathfrak{S}_+(\lambda)}k, \qquad||\mathfrak{S}_-(\lambda)||:=-\sum\limits_{k\in\mathfrak{S}_-(\lambda)}k,$$
and we have used the following facts.
For any partition $\lambda$,
the charge zero condition implies $$|\mathfrak{S}_+(\lambda)|=|\mathfrak{S}_-(\lambda)|
\leq |\lambda|.$$
Further, for any self-conjugate partition $\lambda$,
the self-conjugate condition $m_i=n_i, 1\leq i \leq r(\lambda)$ implies	\[||\mathfrak{S}_+(\lambda)||=||\mathfrak{S}_-(\lambda)||=\frac{|\lambda|}{2}.\] 
Therefore,
for both cases in equations \eqref{eqn:t>1 sigma} and \eqref{eqn:t<1 sigma},
we have the following estimate
\begin{align*}
	\left|\sum\limits_{\lambda\in\mathscr{P}^s}\langle\lambda| q^{H}:T(t): |\lambda^t\rangle\right|
	=\left|\sum\limits_{\lambda\in\mathscr{P}^s}q^{|\lambda|}\sigma_\lambda(t)\right|
	\leq\sum\limits_{\lambda\in\mathscr{P}^s}(q^{\epsilon|\lambda|/2}+|\lambda|q^{|\lambda|}),
\end{align*}
which implies that
\begin{align*}
	\sum\limits_{\lambda\in\mathscr{P}^s}\langle\lambda| q^{H}:T(t): |\lambda^t\rangle
\end{align*}
is absolutely convergent in the region $\Delta_{\epsilon,1}$ since the partition function $Z_s(q)=\sum_{\lambda\in\mathscr{P}^s}q^{|\lambda|}$ in equation \eqref{eqn:Z factor} is absolutely convergent when $|q|<1$.
So $:G(t):$ is a holomorphic function in the region $\triangle_{\epsilon,1}$ as well.
Recall that 
\begin{align*}
	T(t)=\frac{1}{t^{1/2}-t^{-1/2}}+:T(t):,
\end{align*}
so it is natural to consider
\begin{align*} 
	G(t)&=\frac{1}{t^{1/2}-t^{-1/2}}+:G(t):
\end{align*}
as a meromorphic function in the same region $\triangle_{\epsilon,1}$ with the unique simple pole at $t=1$. 
As a consequence,
$G(t)$ admits a meromorphic continuation to the whole complex plane $\mathbb{C}$ by applying the $q$-difference equation \eqref{oneqdiff}, with singularities only at $t=0$ and $t=q^{2m}$, $m\in \mathbb{Z}$.

On the other hand,
denote by $\tilde{G}(t)$ the right hand side of equation \eqref{eqn:one point formula}.
By Lemma \ref{lem:theta q-diff}, $\tilde{G}(t)$ satisfies
\begin{align*}
	\tilde{G}(q^2t)=-\tilde{G}(t),
\end{align*}
which is exactly equivalent to the $q$-difference equation \eqref{oneqdiff} for $G(t)$.
Moreover,
from the infinite product formulas for $\Theta_1(t;q)$ and $\Theta_3(t;q)$ in Lemma \ref{lem:theta inf prod},
$\tilde{G}(t)$ is also a meromorphic function in $\mathbb{C}$ and more explicitly,
\begin{align}\label{eqn:inf prod for tG}
	\begin{split}
	\tilde{G}(t)
	=\frac{1}{t^{1/2}-t^{-1/2}}
	\cdot &\prod_{m=1}^{\infty}\frac{(1-q^{2m})^2}{(1-q^{2m}/t)(1-q^{2m}t)}\\
	&\qquad\qquad\cdot \prod_{m=1}^{\infty}\frac{(1+q^{2m-1}/t)(1+q^{2m-1}t)}{(1+q^{2m-1})^2}.
	\end{split}
\end{align}
The product formula above shows that,
$\tilde{G}(t)$ shares the same singularities of $G(t)$ at $t=q^{2m}$, $m\in \mathbb{Z}$.

Combining all above,
we consider the function
\[G(t)\big/\tilde{G}(t),\]
which is a holomorphic function for $t\in\mathbb{C}\backslash \{0\}$ and 
satisfies the periodic condition
\begin{align*}
	G(q^{-2}t)\big/\tilde{G}(q^{-2}t)
	=G(t)\big/\tilde{G}(t).
\end{align*}
Thus,
it must be a constant function.
Since  $\lim\limits_{t\rightarrow1}G(t)\big/\tilde{G}(t)=1$, we have $G(t)=\tilde{G}(t)$ as desired.
\end{proof}

\begin{ex}
	We expand the one-point function $G(t)$ with respect to $q$ and list the first few of leading terms:
	\begin{align*}
		G(t)
		=&\frac{\sqrt{t}}{t-1}+\frac{ (t-1)}{\sqrt{t}}q
		-\frac{ (t-1)}{\sqrt{t}}q^2
		+\frac{ (t-1) (t+1)^2}{t^{3/2}}q^3
		-\frac{ (t-1)}{\sqrt{t}}q^4\\
		&+\frac{ (t-1) (1 + t^2) (t^2+t+1)}{t^{5/2}}q^5
		-\frac{ (t-1) (t+1)^2}{t^{3/2}}q^6\\
		&+\frac{  (t-1) (t+1)^2 (t^2+1) (t^2-t+1)}{t^{7/2}}q^7
		-\frac{ (t-1)}{\sqrt{t}}q^8\\
		&+\frac{(t-1)(t^8+t^7+t^6+2t^5+3t^4+2t^3+t^2+t+1)}{t^{9/2}}q^9
		+O\left(q^{10}\right).
	\end{align*}
\end{ex}

\subsection{Quasimodularity for the one-point function}	
The closed formula \eqref{eqn:one point formula} for the one-point function $G(t)$ involving theta functions implies that $G(t)$ has certain quasimodularity.
Below,
we give a precise statement and prove the equation \eqref{eqn:G as Gl} in Corollary \ref{cor:main one}.

We first review a few well-known facts on the Eisenstein series.
We refer the readers to \cite{DS,KZ95}.
\begin{defn}\label{def:Gl Gl}
	The $G_\ell(\tau), \ell\in2\mathbb{Z}_{>0}$ is standard Eisenstein series for $\Gamma(1)=SL_2(\mathbb{Z})$ defined by
	\begin{align*}
		G_\ell(\tau):=&-\frac{B_\ell}{2\ell}+\sum\limits_{n=1}^{\infty}\sum\limits_{d|n\atop d>0}(\frac{n}{d})^{\ell-1}e^{2\pi i n\tau},
	\end{align*}
	where $B_\ell$ is the $\ell$-th Bernoulli number.
	The $G^{(1,1)}_\ell(\tau), \ell\in2\mathbb{Z}_{>0}$ is the Eisenstein series for the congruence subgroup $\Gamma(2)$ with index vector $(1,1) \in \mathbb{Z}_2\times\mathbb{Z}_2$ defined by
	\begin{align*}
		G^{(1,1)}_\ell(\tau):=\sum\limits_{n=1}^{\infty}\sum\limits_{d|n,2\nmid d\atop d>0}(-1)^{n/d}(\frac{n}{d})^{\ell-1}e^{\pi i n\tau}.
	\end{align*}
\end{defn}

For even $\ell>2$, $G_\ell(\tau)$ and $G^{(1,1)}_\ell(\tau)$ are modular forms of weight $\ell$ for $\Gamma(1)$ and $\Gamma(2)$ respectively,
while $G_2(\tau)$ and $G^{(1,1)}_2(\tau)$ are quasimodular forms of weight $2$ for $\Gamma(1)$ and $\Gamma(2)$ respectively.
For our application, we regard $G_\ell(\tau), \ell\in2\mathbb{Z}_{>0}$ as a (quasi)modular form for $\Gamma(2)$ as well.

\begin{prop}\label{1dmodular}
	Let $q=e^{\pi i \tau}$ and $t=e^{2\pi i z}$, then
	the one-point function is given by
	\[G(t)=\frac{1}{2\pi i z}\exp\Bigg(\sum\limits_{\ell\in 2\Z_+} 2\left(G_\ell(\tau)- G^{(1,1)}_\ell(\tau)\right)\frac{(2\pi i z)^\ell}{\ell!}\Bigg).\]
\end{prop}
\begin{proof}
	First,
	applying the expansions
	\begin{align*}
		\log(1-q^{2m})=&-\sum\limits_{n=1}^{\infty}\frac{q^{2mn}}{n},\\
		\log(1+q^{2m-1})=&\sum\limits_{n=1}^{\infty}(-1)^{n-1}\frac{q^{(2m-1)n}}{n}
	\end{align*}
	to the formula \eqref{eqn:one point formula} for $G(t)$ in terms of theta functions and using the infinite product formulas in Lemma \ref{lem:theta inf prod},
	we have
	\begin{align}\label{eqn:logG 1}
	\begin{split}
		&\log G(t)-\log\frac{1}{t^{1/2}-t^{-1/2}}\\
		&\qquad=\sum\limits_{m=1}^{\infty}\sum\limits_{n=1}^{\infty}\frac{1}{n}(t^n+t^{-n}-2)(q^{2mn}-(-1)^nq^{(2m-1)n})\\
		&\qquad=\sum\limits_{\ell\in 2\Z_+}\sum\limits_{m=1}^{\infty}\sum\limits_{n=1}^{\infty}2n^{\ell-1}
		\big(e^{2mn\pi i \tau}-(-1)^ne^{(2m-1)n\pi i \tau}\big)\frac{(2\pi iz)^\ell}{\ell!}.
	\end{split}
	\end{align}
	
Denote the coefficient of $z^\ell/\ell!$ in the last line of the equation above by
\begin{align}\label{eqn:def mbG}
	\begin{split}
	H_\ell(\tau)
	&:=\sum\limits_{m=1}^{\infty}\sum\limits_{n=1}^{\infty} n^{\ell-1}\big(e^{2mn\pi i \tau}-(-1)^ne^{(2m-1)n\pi i \tau}\big)\\
	&=\sum\limits_{n=1}^{\infty}\sum\limits_{d|n\atop d>0}\left(\frac{n}{d}\right)^{\ell-1}e^{2\pi i n\tau}-\sum\limits_{n=1}^{\infty}\sum\limits_{d|n,2\nmid d\atop d>0}(-1)^{n/d}\left(\frac{n}{d}\right)^{\ell-1}e^{\pi i n\tau}.
	\end{split}
\end{align}
Then the formula \eqref{eqn:logG 1} can be rewritten as
\begin{align}\label{eqn:G as mbG}
	G(t)=\frac{1}{e^{\pi i z}-e^{-\pi i z}}
	\exp\Bigg(\sum\limits_{\ell\in 2\Z_+} 2H_\ell(\tau)\frac{(2\pi iz)^\ell}{\ell!}\Bigg).
\end{align}
From the Definition \ref{def:Gl Gl} for $G_\ell(\tau)$ and $G^{(1,1)}_\ell(\tau)$,
the function $H_\ell(\tau)$ defined in equation \eqref{eqn:def mbG} is also given by
\[H_\ell(\tau)= G_\ell(\tau)+\frac{B_\ell}{2\ell}-G^{(1,1)}_\ell(\tau).\]
Moreover,
combining the following identity for Bernoulli numbers
\begin{align}\label{eqn:Bernoulli}
	\exp\Bigg(\sum\limits_{\ell\in 2\Z_+}\frac{B_\ell}{\ell}\cdot\frac{(2\pi i z)^\ell}{\ell!}\Bigg)
	=\frac{e^{\pi i z}-e^{-\pi i z}}{2\pi i z},
\end{align}
formula \eqref{eqn:G as mbG} is reduced to
\[G(t)=\frac{1}{2\pi i z}\exp\Bigg(\sum\limits_{\ell\in 2\Z_+} 2\left(G_\ell(\tau)- G^{(1,1)}_\ell(\tau)\right)\frac{(2\pi i z)^\ell}{\ell!}\Bigg).\]
\end{proof}

\subsection{An explicit formula for the two-point function}
In this subsection,
we derive an explicit formula for the two-point function
\begin{align}\label{eqn:two point}
	G(t_1,t_2)&=\left\langle \sum\limits_{i=1}^{\infty} t_1^{\lambda_i-i+\frac{1}{2}}
	\cdot \sum\limits_{i=1}^{\infty} t_2^{\lambda_i-i+\frac{1}{2}} \right\rangle^s_{q},
\end{align}
which proves the equation \eqref{eqn:G2 main} in Corollary \ref{cor:main one}.

Recall that,
Theorem \ref{main} gives the following $q$-difference equation for the two-point function $G(t_1,t_2)$:
\begin{align}\label{eqn:q-diff two}
	G(q^{-2}t_1,t_2)=-G(t_1,t_2)+G(t_1/t_2)-G(t_1t_2).
\end{align}

The following $q$-difference equations for quotients of theta functions and their derivatives are useful to analyze the two-point function $G(t_1,t_2)$.
\begin{lem}\label{lem:theta q-diff2}
	We have
	\begin{align*}
		\frac{\Theta_3(q^{-2}t_1/t_2;q)}{\Theta_1(q^{-2}t_1/t_2;q)}&=-\frac{\Theta_3(t_1/t_2;q)}{\Theta_1(t_1/t_2;q)},\\
		\frac{\Theta^{'}_3(q^{-2}t_1t_2;q)}{\Theta_1(q^{-2}t_1t_2;q)}&=-\frac{\Theta^{'}_3(t_1t_2;q)}{\Theta_1(t_1t_2;q)}-\frac{\Theta^3(t_1t_2;q)}{\Theta_1(t_1t_2;q)},\\
		\frac{\Theta^{'}_1(q^{-2}t_1;q)}{\Theta_1(q^{-2}t_1;q)}&=\frac{\Theta^{'}_1(t_1;q)}{\Theta_1(t_1;q)}+1.
	\end{align*}
\end{lem}
\begin{proof}
	These equations can be directly proved by using Lemma \ref{lem:theta q-diff} and Corollary \ref{cor:theta' q-diff}.
\end{proof}

\begin{prop}\label{thm:two-point formula}
	The two-point function $G(t_1,t_2)$ admits a meromorphic continuation to $(t_1,t_2)\in\mathbb{C}^2$
	and more explicitly,
	it is given by
	\begin{align}\label{eqn:G2}
		\begin{aligned}
			G(t_1,t_2)&=q^{1/4}\prod_{m=1}^{\infty}\frac{(1-q^{2m})^2}{(1+q^{2m-1})^2}
			\cdot\bigg[\frac{\Theta^{'}_3(t_1t_2;q)}{\Theta_1(t_1t_2;q)}\\
			&\qquad-\frac{\Theta^{'}_1(t_1;q)}{\Theta_1(t_1;q)}\cdot\frac{\Theta_3(t_1/t_2;q)}{\Theta_1(t_1/t_2;q)}-\frac{\Theta^{'}_1(t_2;q)}{\Theta_1(t_2;q)}\cdot\frac{\Theta_3(t_2/t_1;q)}{\Theta_1(t_2/t_1;q)}\bigg].
		\end{aligned}
	\end{align}
\end{prop}
\begin{proof}
Denote the two-point function with normal ordering as
\begin{align*}
	:G(t_1,t_2):=\frac{1}{\sum\limits_{\lambda\in\mathscr{P}^s} \langle\lambda|q^{H}|\lambda^t\rangle}
	\cdot \sum\limits_{\lambda\in\mathscr{P}^s}\langle\lambda| q^{H}:T(t_1): :T(t_2): |\lambda^t\rangle.
\end{align*}
Similar to the one-point function case,
we first show the analyticity of $:G(t_1,t_2):$ in the following region
\begin{align*}
	\Delta_{\epsilon,2}:=\big\{(t_1,t_2)\in\mathbb{C}^2\big|
	\ q^{2-2\epsilon}<|t_1|< q^{-2+2\epsilon},
	1<|t_2|< q^{-\epsilon}\big\},
\end{align*}
where $\epsilon$ is a small positive number. We only need to estimate 
\begin{align*}
	q^H:T(t_1): :T(t_2): |\lambda\rangle
	= q^{|\lambda|}\sigma_\lambda(t_1) \sigma_\lambda (t_2)
	\cdot |\lambda\rangle,
\end{align*}
where $\sigma_\lambda(t):=\sum_{k\in\mathfrak{S}_(\lambda) +(\lambda)}t^k-\sum_{k\in\mathfrak{S}_-(\lambda)}t^k$.
This will split into two cases. For $|t_1|\geq1$,
we have
\begin{align}\label{eqn:estone}
	\begin{split}
		|q^{|\lambda|}\sigma_\lambda(t_1) \sigma_\lambda (t_2)|
		\leq& q^{|\lambda|}\big(|t_1|^{\|\mathfrak{S}_+(\lambda)\|} +|\mathfrak{S}_-(\lambda)|\big)
		\cdot\big(|t_2|^{\|\mathfrak{S}_+(\lambda)\|} +|\mathfrak{S}_-(\lambda)|\big)\\
		\leq&(|q|^{\epsilon|\lambda|}+|\lambda|q^{|\lambda|})
		\cdot(|q|^{-\epsilon|\lambda|/2}+|\lambda|)\\
		=&|q|^{\epsilon|\lambda|/2}+|\lambda||q|^{(1-\epsilon/2)|\lambda|}+|\lambda||q|^{\epsilon|\lambda|}+|\lambda|^2|q|^{|\lambda|},
	\end{split}	
\end{align}
where we have used the facts $|\mathfrak{S}_-(\lambda)|\leq|\lambda|$
and $\|\mathfrak{S}_+(\lambda)\|=|\lambda|/2$
for the self-conjugate partition $\lambda$.
Similarly,
for $|t_1|<1$,
we have
\begin{align}\label{eqn:esttwo}
	\begin{split}
		|q^{|\lambda|}\sigma_\lambda(t_1) \sigma_\lambda (t_2)|
		\leq& q^{|\lambda|}\big(|t_1|^{-\|\mathfrak{S}_-(\lambda)\|} +|\mathfrak{S}_+(\lambda)|\big)
		\cdot\big(|t_2|^{\|\mathfrak{S}_+(\lambda)\|} +|\mathfrak{S}_-(\lambda)|\big)\\
		\leq&(|q|^{\epsilon|\lambda|}+|\lambda|q^{|\lambda|})
		\cdot(|q|^{-\epsilon|\lambda|/2}+|\lambda|)\\
		=&|q|^{\epsilon|\lambda|/2}+|\lambda||q|^{(1-\epsilon/2)|\lambda|}+|\lambda||q|^{\epsilon|\lambda|}+|\lambda|^2|q|^{|\lambda|}.
	\end{split}
\end{align}
As a consequence, by combining equations \eqref{eqn:estone} and \eqref{eqn:esttwo}, 
the series
\begin{align*}
	\sum\limits_{\lambda\in\mathscr{P}^s}\langle\lambda| q^{H}:T(t_1): :T(t_2): |\lambda^t\rangle
\end{align*}
is absolutely convergent in the region $\Delta_{\epsilon,2}$,
which implies the analyticity of $:G(t_1,t_2):$ in the same region.
Recall that 
\begin{align*}
	T(t)=&\frac{1}{t^{1/2}-t^{-1/2}}+:T(t):,
\end{align*}
it follows that $G(t_1,t_2)$ could be considered as a meromorphic function in the region $\Delta_{\epsilon,2}$ since
\begin{align}\label{eqn:G2 t=1}
	\begin{split}
		G(t_1,t_2)=\frac{1}{(t_1^{1/2}-t_1^{-1/2})(t_2^{1/2}-t_2^{-1/2})}&\\
		+\frac{:G(t_1):}{t_2^{1/2}-t_2^{-1/2}}&+\frac{:G(t_2):}{t_1^{1/2}-t_1^{-1/2}}+:G(t_1,t_2):.
	\end{split}	
\end{align}
Here we can see that the only singularity of $G(t_1,t_2)$ in $\Delta_{\epsilon,2}$ is divisor $\{t_1=1\}$. 
Furthermore, in $\Delta_{\epsilon,2}$,
\begin{align}\label{eqn:m0}
	G(t_1,t_2)=\frac{G(t_2)}{t_1^{1/2}-t_1^{-1/2}}+(\text{regular on\ }\{t_1=1\}).
\end{align}

By the $q$-difference equation \eqref{eqn:q-diff two},
the two-point function $G(t_1,t_2)$ could be extended to a meromorphic function on the following larger region
\begin{align*}
	\widetilde{\Delta_{\epsilon,2}}:=\{(t_1,t_2)\in\mathbb{C}^2\ |\ 
	t_1\in\mathbb{C}\setminus \{0\},
	\quad 1<|t_2|< q^{-\epsilon}\}.
\end{align*}
As a consequence,
$G(t_1,t_2)$ is singular at
\begin{align}
	\begin{split}
		\{q^{2m}t_1=1|\ m\in \mathbb{Z}\}
		&\cup \big\{q^{2m}t_1t_2=1 |\ m\in \mathbb{Z}\setminus\{0\}\big\}\\ 
		&\cup \big\{q^{2m}t_1/t_2=1 |\ m\in \mathbb{Z}\setminus\{0\}\big\}
	\end{split}	
\end{align}
and has no other singularity in $\widetilde{\Delta_{\epsilon,2}}$.
Below,
we explain this in detail.

After multiplying by $q^{2m}$, the region $\Delta_{\epsilon,2}$ is translated to a new region denoted by
\begin{align}
	\Delta_{\epsilon,2}^m:=\{(t_1,t_2)\in\mathbb{C}^2\ |\ 
	q^{2m+2-2\epsilon}<|t_1|<q^{2m-2+2\epsilon},
	1<|t_2|< q^{-\epsilon}\}.
\end{align}
For the first case of $m=1$,
by applying the $q$-difference equation \eqref{eqn:q-diff two} for $G(t_1,t_2)$,
we have
\begin{align*}
	G(t_1,t_2)&=-G(q^{-2}t_1,t_2)-G(t_1/t_2)+G(t_1t_2)\\
	&=-\frac{G(t_2)}{(q^{-2}t_1)^{1/2}-(q^{-2}t_1)^{-1/2}}+(\text{regular on\ }\{q^{-2}t_1=1\})
\end{align*}
around $\{q^{-2}t_1=1\}$.
However,
it is worth noting that $\{q^{-2}t_1=1\}$ is not the only singular locus of $G(t_1,t_2)$ in this region $\Delta_{\epsilon,2}^1$,
since $G(t_1/t_2)$ and $G(t_1t_2)$ are singular at some other loci.
Indeed,
by the singularity of the one-point function analyzed in the proof of Proposition \ref{prop:1-point},
we have
\begin{align*}
	G(t_1,t_2)&=-G(q^{-2}t_1,t_2)-G(t_1/t_2)+G(t_1t_2)\\
	&=\frac{1}{(q^{-2}t_1/t_2)^{1/2}-(q^{-2}t_1/t_2)^{-1/2}}+(\text{regular on\ }\{q^{-2}t_1/t_2=1\})
\end{align*}
and
\begin{align*}
	G(t_1,t_2)
	=-\frac{1}{(q^{-2}t_1t_2)^{1/2}-(q^{-2}t_1t_2)^{-1/2}}+(\text{regular on\ }\{q^{-2}t_1t_2=1\})
\end{align*}
around $\{q^{-2}t_1/t_2=1\}$ and $\{q^{-2}t_1t_2=1\}$, respectively.
The other regions $\Delta_{\epsilon,2}^m$ for $m\neq0,1$ can be analyzed in a similar method.
More precisely,
for the region $(t_1,t_2)\in \Delta_{\epsilon,2}^m$ with $m>0$,
inductively using the $q$-difference equation \eqref{eqn:q-diff two} for $G(t_1,t_2)$ and the $q$-difference equation \eqref{oneqdiff} for $G(t)$, we have
\begin{align}
	\begin{split}
	G(t_1,t_2)&=-G(q^{-2}t_1,t_2)-G(t_1/t_2)+G(t_1t_2)\\
	&=(-1)^m\big(G(q^{-2m}t_1,t_2)-mG(t_1/t_2)+mG(t_1t_2)\big).
	\end{split}
\end{align}
Thus,
we obtain that,
in the region $\Delta_{\epsilon,2}^m$ for $m>0$,
$G(t_1,t_2)$ has three singular loci and more precisely,
\begin{align}
	G(t_1,t_2)
	&=\frac{(-1)^mG(t_2)}{(q^{-2m}t_1)^{1/2}-(q^{-2m}t_1)^{-1/2}}
	+(\text{regular on\ }\{q^{-2m}t_1=1\}) \label{eqn:m1}\\
	&=\frac{(-1)^{m-1}m}{(q^{-2m}t_1/t_2)^{1/2}-(q^{-2m}t_1/t_2)^{-1/2}}+(\text{regular on\ }\{q^{-2m}t_1/t_2=1\}) \label{eqn:m2}\\
	&=\frac{(-1)^{m}m}{(q^{-2m}t_1t_2)^{1/2}-(q^{-2m}t_1t_2)^{-1/2}}+(\text{regular on\ }\{q^{-2m}t_1t_2=1\}). \label{eqn:m3}
\end{align}
The singular locus in the region $\Delta_{\epsilon,2}^m$ with $m<0$ can be analyzed by the same way.
Then one has equations \eqref{eqn:m1}, \eqref{eqn:m2}, and \eqref{eqn:m3} hold for all $m\in\mathbb{Z}\setminus\{0\}$.

On the other hand,
denote by $\tilde{G}(t_1,t_2)$ the right hand side of equation \eqref{eqn:G2}. We are going to show that $\tilde{G}(t_1,t_2)$ also satisfies the $q$-difference equation \eqref{eqn:q-diff two} and has the same singularities of $G(t_1,t_2)$ in $\widetilde{\Delta_{\epsilon,2}}$.  

From the Lemma \ref{lem:theta q-diff2},
we have the following $q$-difference equation
\begin{align}
	\begin{split}
	\tilde{G}(q^{-2}t_1,t_2)&=q^{1/4}\prod_{m=1}^{\infty}\frac{(1-q^{2m})^2}{(1+q^{2m-1})^2}
	\cdot\bigg[\frac{\Theta^{'}_3(q^{-2}t_1t_2;q)}{\Theta_1(q^{-2}t_1t_2;q)}-\\
	&\qquad\frac{\Theta^{'}_1(q^{-2}t_1;q)}{\Theta_1(q^{-2}t_1;q)}\cdot\frac{\Theta_3(q^{-2}t_1/t_2;q)}{\Theta_1(q^{-2}t_1/t_2;q)}-\frac{\Theta^{'}_1(t_2;q)}{\Theta_1(t_2;q)}\cdot\frac{\Theta_3(t_2/q^{-2}t_1;q)}{\Theta_1(t_2/q^{-2}t_1;q)}\bigg]\\
	&=-\tilde{G}(t_1,t_2)+G(t_1/t_2)-G(t_1t_2)
	\end{split}
\end{align}
for the $\tilde{G}(t_1,t_2)$,
which is equivalent to the $q$-differential equation \eqref{eqn:q-diff two} satisfied by $G(t_1,t_2)$ as desired.

Now we analyze the singularities of $\tilde{G}(t_1,t_2)$ in $\widetilde{\Delta_{\epsilon,2}}$. We start from converting the formula of $\tilde{G}(t_1,t_2)$ into the following form:
\begin{align}\label{eqn:G2fixed}
	\begin{aligned}
		\tilde{G}(t_1,t_2)&=q^{1/4}\prod_{m=1}^{\infty}\frac{(1-q^{2m})^2}{(1+q^{2m-1})^2}
		\cdot\bigg[\frac{\Theta^{'}_3(t_1t_2;q)}{\Theta_3(t_1t_2;q)}\cdot\frac{\Theta_3(t_1t_2;q)}{\Theta_1(t_1t_2;q)}-\\
		&\ \qquad\frac{\Theta^{'}_1(t_1;q)}{\Theta_1(t_1;q)}\cdot\frac{\Theta_3(t_1/t_2;q)}{\Theta_1(t_1/t_2;q)}-\frac{\Theta^{'}_1(t_2;q)}{\Theta_1(t_2;q)}\cdot\frac{\Theta_3(t_2/t_1;q)}{\Theta_1(t_2/t_1;q)}\bigg],
	\end{aligned}
\end{align}
where the derivative terms can be expanded as follows,
\begin{align}\label{thetader1}
	\begin{split}
		\frac{\Theta_1^{'}(t;q)}{\Theta_1(t;q)}&=t\frac{\partial}{\partial t}\log\Theta_1(t;q)\\
		&=\frac{1}{2}\cdot\frac{t^{1/2}+t^{-1/2}}{t^{1/2}-t^{-1/2}}-\sum\limits_{m=1}^\infty\frac{tq^{2m}}{1-tq^{2m}}+\sum\limits_{m=1}^\infty\frac{t^{-1}q^{2m}}{1-t^{-1}q^{2m}},
	\end{split}
\end{align}
and
\begin{align}\label{thetader3}
	\frac{\Theta_3^{'}(t;q)}{\Theta_3(t;q)}&=t\frac{\partial}{\partial t}\log\Theta_3(t;q)=\sum\limits_{m=1}^\infty\frac{tq^{2m-1}}{1+tq^{2m-1}}-\sum\limits_{m=1}^\infty\frac{t^{-1}q^{2m-1}}{1+t^{-1}q^{2m-1}}.
\end{align}
By Lemma \ref{lem:theta inf prod}, we can see that $\tilde{G}(t_1,t_2)$ is meromorphic in $\mathbb{C}^2$ and its singularities are only contributed by factors in denominators. Especially in $\widetilde{\Delta_{\epsilon,2}}$, the singularities of $\tilde{G}(t_1,t_2)$ are only located at
\begin{align*}
	\{q^{2m}t_1=1|\ m\in \mathbb{Z}\}
	\cup \big\{q^{2m}t_1t_2=1 |\ m\in \mathbb{Z}\setminus\{0\}\big\}.
\end{align*}
Below, we confirm these singularities by computing the residues of $\tilde{G}(t_1,t_2)$ along these loci.

First, for the divisor $\{t_1=1\}$ in $\Delta_{\epsilon,2}$, we have
\begin{align}
	\begin{split}
	\lim\limits_{t_1\rightarrow 1}(t_1^{1/2}-t_1^{-1/2})\tilde{G}(t_1,t_2)&=-q^{1/4}\prod_{m=1}^{\infty}\frac{(1-q^{2m})^2}{(1+q^{2m-1})^2}
	\cdot\frac{\Theta_3(1/t_2;q)}{\Theta_1(1/t_2;q)}\\
	&=q^{1/4}\prod_{m=1}^{\infty}\frac{(1-q^{2m})^2}{(1+q^{2m-1})^2}
	\cdot\frac{\Theta_3(t_2;q)}{\Theta_1(t_2;q)}\\
	&=G(t_2),
	\end{split}
\end{align}
where we have applied the following equation (see the infinite product formulas in Lemma \ref{lem:theta inf prod})
\begin{align*}
	\frac{\Theta_3(1/t;q)}{\Theta_1(1/t;q)}=-\frac{\Theta_3(t;q)}{\Theta_1(t;q)}.
\end{align*}
Thus in the region $\Delta_{\epsilon,2}$,
\begin{align}\label{eqn:tG at 10}
		\tilde{G}(t_1,t_2)=\frac{G(t_2)}{t_1^{1/2}-t_1^{-1/2}}+(\text{regular on\ }\{t_1=1\}),
\end{align}
which matches the singularity of $G(t_1,t_2)$ shown in equation \eqref{eqn:m0}.
Generally, for the divisor $\{q^{-2m}t_1=1\}$ in the region $\Delta_{\epsilon,2}^m$ with $m\in \mathbb{Z}\setminus\{0\}$,
we have
\begin{align}
	\begin{split}
	\lim\limits_{t_1\rightarrow q^{2m}}&
	\big((q^{-2m}t_1)^{1/2}-(q^{-2m}t_1)^{-1/2}\big)\tilde{G}(t_1,t_2)\\
	=&-q^{1/4}\prod_{m=1}^{\infty}\frac{(1-q^{2m})^2}{(1+q^{2m-1})^2}\cdot\frac{\Theta_3(q^{2m}/t_2;q)}{\Theta_1(q^{2m}/t_2;q)}\\
	=&q^{1/4}\prod_{m=1}^{\infty}\frac{(1-q^{2m})^2}{(1+q^{2m-1})^2}\cdot(-1)^{m}\frac{\Theta_3(t_2;q)}{\Theta_1(t_2;q)}
	=(-1)^{m}G(t_2).
	\end{split}
\end{align}
As a result, for $(t_1,t_2)\in \Delta_{\epsilon,2}^m$ with all $m\in\mathbb{Z}\setminus\{0\}$, we have
\begin{align}\label{eqn:tG at 1}
	\tilde{G}(t_1,t_2)
	=\frac{(-1)^mG(t_2)}{(q^{-2m}t_1)^{1/2}-(q^{-2m}t_1)^{-1/2}}
	+(\text{regular on\ }\{q^{-2m}t_1=1\}),
\end{align}
which matches the singularity of $G(t_1,t_2)$ shown in equation \eqref{eqn:m1}.

We continue to deal with the divisors $\{q^{-2m}t_1t_2=1\}$ and $\{q^{-2m}t_1/t_2=1\}$ with $m\in\mathbb{Z}\setminus\{0\}$.
For the divisor $\{q^{-2m}t_1t_2=1\}$,
let $u=t_1/t_2$.
Then by using equation \eqref{thetader1}, we have
\begin{align}
	\begin{split}
	&\lim\limits_{u\rightarrow q^{2m}}
	\big((q^{-2m}u)^{1/2}-(q^{-2m}u)^{-1/2}\big)\tilde{G}(ut_2,t_2)\\
	=&q^{1/4}\prod_{m=1}^{\infty}\frac{(1-q^{2m})^2}{(1+q^{2m-1})^2}\cdot\lim\limits_{u\rightarrow q^{2m}}\big((q^{-2m}u)^{1/2}-(q^{-2m}u)^{-1/2}\big)\\
	&\qquad\cdot\left(-\frac{\Theta_1^{'}(ut_2;q)}{\Theta_1(ut_2;q)}\cdot\frac{\Theta_3(u;q)}{\Theta_1(u;q)}-\frac{\Theta_1^{'}(t_2;q)}{\Theta_1(t_2;q)}\cdot\frac{\Theta_3(1/u;q)}{\Theta_1(1/u;q)}\right)\\
	=&(-1)^{m}\left(-\frac{\Theta_1^{'}(q^{2m}t_2;q)}{\Theta_1(q^{2m}t_2;q)}+\frac{\Theta_1^{'}(t_2;q)}{\Theta_1(t_2;q)}\right)\\
	=&(-1)^{m-1}m.
	\end{split}
\end{align}
For the divisor $\{q^{-2m}t_1/t_2=1\}$,
let $\tilde{u}=t_1t_2$.
Then, similarly, by using equation \eqref{thetader3},
\begin{align}
	\begin{split}
	&\lim\limits_{\tilde{u}\rightarrow q^{2m}}
	\big((q^{-2m}\tilde{u})^{1/2}-(q^{-2m}\tilde{u})^{-1/2}\big)\tilde{G}(t_1,t_2)\\
	=&q^{1/4}\prod_{m=1}^{\infty}\frac{(1-q^{2m})^2}{(1+q^{2m-1})^2}\cdot\lim\limits_{\tilde{u}\rightarrow q^{2m}}\big((q^{-2m}\tilde{u})^{1/2}-(q^{-2m}\tilde{u})^{-1/2}\big)
	\frac{\Theta_3^{'}(\tilde{u};q)}{\Theta_1(\tilde{u};q)}\\
	=&q^{1/4}\prod_{m=1}^{\infty}\frac{(1-q^{2m})^2}{(1+q^{2m-1})^2}\cdot\lim\limits_{\tilde{u}\rightarrow q^{2m}}\big((q^{-2m}\tilde{u})^{1/2}-(q^{-2m}\tilde{u})^{-1/2}\big)
	\frac{\Theta_3^{'}(\tilde{u};q)}{\Theta_3(\tilde{u};q)}\frac{\Theta_3(\tilde{u};q)}{\Theta_1(\tilde{u};q)}\\
	=&(-1)^{m}m.
	\end{split}
\end{align}
In conclusion, for $(t_1,t_2)\in \Delta_{\epsilon,2}^m$ with $m\in\mathbb{Z}\setminus\{0\}$, we have
\begin{align}
	\tilde{G}(t_1,t_2)
	&=\frac{(-1)^{m-1}m}{(q^{-2m}t_1/t_2)^{1/2}-(q^{-2m}t_1/t_2)^{-1/2}}+(\text{regular on\ }\{q^{-2m}t_1/t_2=1\}) \label{eqn:tG at 2}\\
	&=\frac{(-1)^{m}m}{(q^{-2m}t_1t_2)^{1/2}-(q^{-2m}t_1t_2)^{-1/2}}+(\text{regular on\ }\{q^{-2m}t_1t_2=1\}), \label{eqn:tG at 3}
\end{align}
which matches the singularities of $G(t_1,t_2)$ shown in equations \eqref{eqn:m2} and \eqref{eqn:m3}, respectively.

By combining equations \eqref{eqn:tG at 10}, \eqref{eqn:tG at 1}, \eqref{eqn:tG at 2} and \eqref{eqn:tG at 3} and comparing them with equations \eqref{eqn:m0}, \eqref{eqn:m1}, \eqref{eqn:m2} and \eqref{eqn:m3},
we obtain that $\tilde{G}(t_1,t_2)$ and $G(t_1,t_2)$ have the same singularities in the region $\widetilde{\Delta_{\epsilon,2}}$.

Finally, we have shown that  both $G(t_1,t_2)$ and $\tilde{G}(t_1,t_2)$ satisfy the same $q$-differential equation \eqref{eqn:q-diff two} and have the same singularities in $\widetilde{\Delta_{\epsilon,2}}$. Let
\begin{align*}
	\mathscr{G}(t_1,t_2):=G(t_1,t_2)-\tilde{G}(t_1,t_2),
\end{align*}
then $\mathscr{G}(t_1,t_2)$ is holomorphic in $\widetilde{\Delta_{\epsilon,2}}$ and satisfies
\begin{align}\label{tworecusion}
	\mathscr{G}(t_1,t_2)=-\mathscr{G}(q^{-2}t_1,t_2).
\end{align}
Hence, when fixing a value of $t_2$, $\mathscr{G}(t_1,t_2)$ is a constant function with respect to $t_1\in \C\setminus\{0\}$. Let $\mathscr{G}(t_1,t_2)\equiv a(t_2)$, then $a(t_2)=-a(t_2)$ by equation \eqref{tworecusion}, which implies that $a(t_2)\equiv0$, that is, $\mathscr{G}(t_1,t_2)\equiv 0$. This establishes $G(t_1,t_2)=\tilde{G}(t_1,t_2)$ in $\widetilde{\Delta_{\epsilon,2}}$. Since $\tilde{G}(t_1,t_2)$ is meromorphic in $\mathbb{C}^2$, it can be regarded as the meromorphic continuation of $G(t_1,t_2)$ onto $\mathbb{C}^2$. This finishes the proof.
\end{proof}

\begin{ex}
We expand two-point function $G(t_1,t_2)$ with respect to $q$ and list the first few of leading terms:
\begin{align*}
	G(t_1,t_2)
	=&\frac{\sqrt{t_1t_2}}{(t_1-1)(t_2-1)}
	+\frac{(t_1^2-t_1+1) \left(t_2^2-t_2+1\right)-t_1t_2}{\sqrt{t_1t_2}(t_1-1)(t_2-1)}(q-q^2)\\
	&+\frac{(t_1^4-t_1^3+t_1^2-t_1+1)(t_2^4-t_2^3+t_2^2-t_2+1)}
	{t_1^{3/2}t_2^{3/2}(t_1-1)(t_2-1)}q^3\\
	&+\frac{t_1t_2(t_1^2-t_1+1) \left(t_2^2-t_2+1\right)-2t_1^2t_2^2}
	{t_1^{3/2}t_2^{3/2}(t_1-1)(t_2-1)}q^3+O(q^4).
\end{align*}
\end{ex}

\section{The quasimodularity of $n$-point function}\label{sec:quasimo}

In this section,
motivated by the result in \cite{BO,EOP,KZ95,Z16} and the explicit formulas in Corollary \ref{cor:main one},
we study the quasimodularity of the $n$-point functions $G(t_1,t_2,...,t_n)$ of the self-conjugate partitions.
We shall prove Theorem \ref{thm:main quasimod}.

The following lemma will be useful when proving the quasimodularity of the correlation function of self-conjugate partitions.
\begin{lem}
	For any positive even integer $\ell$,
	the series
	\begin{align}\label{eqn:def bbG}
		\mathbb{G}_{\ell}(\tau):=(1-2^{\ell-1})\zeta(1-\ell)/2+\sum\limits_{n=1}^{\infty}\sum\limits_{d|n,2\nmid d\atop d>0}(-1)^{n}d^{\ell-1}e^{\pi i n\tau}
	\end{align}
	is a quasimodular form of weight $\ell$ for the congruence subgroup $\Gamma(2)$.
\end{lem}
\begin{proof}
	In Definition \ref{def:Gl Gl}, we introduced the Eisenstein series $G_\ell(\tau)$ for $\Gamma(1)$ and $G^{(1,1)}_\ell(\tau)$ for $\Gamma(2)$. Here we need another two Standard Eisenstein series of weight $\ell$ for $\Gamma(2)$ and refer the readers to \cite{DS} (Our definition and notation are slightly different from those in \cite{DS}, and they are equivalent up to constants) :
	\begin{align*}
		&G^{(1,0)}_\ell(\tau):=\sum\limits_{n=1}^{\infty}\sum\limits_{d|n,2\nmid d\atop d>0}(\frac{n}{d})^{\ell-1}e^{\pi i n\tau},\\
		&G^{(0,1)}_\ell(\tau):=(2^{\ell}-1)\zeta(1-\ell)/2+\sum\limits_{n=1}^{\infty}\sum\limits_{d|n,2|d\atop d>0}(-1)^{n/d}(\frac{n}{d})^{\ell-1}e^{\pi i n\tau}.
	\end{align*}
	Notice that $G_\ell(\tau)$ can also be regarded as a quasimodular form for $\Gamma(2)$ and 
	\begin{align*}
		&G_\ell(\tau)=\frac{1}{2^\ell-1}\big(G^{(1,0)}_\ell(\tau)+G^{(0,1)}_\ell(\tau)+G^{(1,1)}_\ell(\tau)\big),\\
		&\mathbb{G}_{\ell}(\tau)=(1-2^{\ell-1})G_\ell(\tau)+G^{(1,1)}_\ell(\tau).
	\end{align*}
	Therefore, $\mathbb{G}_{\ell}(\tau)$ is quasimodular of weight $\ell$ for $\Gamma(2)$.
\end{proof}

Recall that the $n$-point function $G(t_1,t_2,...,t_n)$ is given by the expectation of the function 
$\prod_{j=1}^n\sum_{i=1}^{\infty} t_j^{\lambda_i-i+\frac{1}{2}}$
on $\mathscr{P}^s$. Now, let $t=e^{2\pi i z}$ and we can expand the function $\sum_{i=1}^{\infty} t^{\lambda_i-i+\frac{1}{2}}$ in the following way:
\begin{equation}\label{expanofG}
\begin{aligned}
	\sum\limits_{i=1}^{\infty} t^{\lambda_i-i+\frac{1}{2}}
	&=\sum\limits_{s\in\mathfrak{S}_+(\lambda)}e^{2\pi i zs}-\sum\limits_{s\in\mathfrak{S}_-(\lambda)}e^{2\pi i zs}+\frac{1}{2\sinh(\pi i z)}\\
	&=\sum\limits_{\ell\geq0} Q_\ell(\lambda) (2\pi i z)^{\ell-1},
\end{aligned}
\end{equation}
which defines a series of functions $Q_\ell: \mathscr{P}^s \rightarrow \Q$ for $\ell\in\mathbb{Z}_{\geq0}$
(see also equation (17) in \cite{Z16}).  
Furthermore, the explicitly formula of $Q_\ell(\lambda)$ is
\begin{align*}
	Q_\ell(\lambda)&=\frac{1}{(\ell-1)!}\sum\limits_{i=1}^{r(\lambda)}\left[\big(m_i+\frac{1}{2}\big)^{\ell-1}-\big(-n_i-\frac{1}{2}\big)^{\ell-1}\right]+\beta_\ell
\end{align*}
for $\ell>0$ and $Q_0(\lambda)=1$,
where $(m_1,...,m_{r(\lambda)}|n_1,...,n_{r(\lambda)})$ is the Frobenius notation of $\lambda$, $r(\lambda)$ is its Frobenius length, and $\beta_\ell$ is given by the series expansion  
\begin{align*}
	\sum\limits_{\ell=0}^\infty\beta_\ell x^\ell=&\frac{x/2}{\sinh(x/2)}=\sum\limits_{n=0}^\infty\frac{(1/2^{2n-1}-1)B_{2n}x^{2n}}{(2n)!},
	\qquad 0<|x|<\pi.
\end{align*}
Moreover, 
it is obvious that,
since we only consider $\lambda \in \mathscr{P}^s$, $Q_{\ell}(\lambda)$ is nonzero only for even $\ell\in\mathbb{Z}_{>0}$.

\begin{thm}[=Theorem \ref{thm:main quasimod}]
	\label{thm:mor}
	Let $t_i=e^{2\pi i z_i}, i=1,2,\cdots,n$,
	and $q=e^{\pi i \tau}$.
	The $n$-point function $G(t_1,t_2,\cdots,t_n)$ of the self-conjugate partitions has the following expansion
	\begin{align*}
		G(t_1,t_2,\cdots,t_n)
		=\sum\limits_{\ell_1,\ell_2,\cdots,\ell_n\geq0}\langle Q_{\ell_1}Q_{\ell_2}\cdots Q_{\ell_n} \rangle^s_q
		\cdot \prod_{j=1}^n (2\pi i z_j)^{\ell_j-1}.
	\end{align*}
	Then for any non-negative integers $\ell_1,...,\ell_n$,
	$\langle Q_{\ell_1}Q_{\ell_2}\cdots Q_{\ell_n} \rangle^s_q$
	is a quasimodular form of weight $\sum_{i=1}^n\ell_i$ for the congruence subgroup $\Gamma(2)$.
\end{thm}
\begin{proof}
	The expansion formula of $G(t_1,t_2,\cdots,t_n)$ is directly from \eqref{expanofG}.
	Following the notation in \cite{Z16},
	we introduce the function $P_{\ell}(\cdot), \ell\in\mathbb{Z}_{\geq0}$, on the set of self-conjugate partitions as
	\begin{align}\label{eqn:def P_l}
		P_\ell(\lambda)&:=\sum\limits_{i=1}^{r(\lambda)}\left[\big(m_i+\frac{1}{2}\big)^{\ell}-\big(-n_i-\frac{1}{2}\big)^{\ell}\right].
	\end{align}
	The relation of $Q_{\ell}(\cdot)$ and $P_\ell(\cdot)$ is
	$$Q_\ell(\lambda)=\frac{P_{\ell-1}(\lambda)}{(\ell-1)!}+\beta_\ell,
	\qquad \ell\in\mathbb{Z}_{>0}.$$
	And $P_\ell(\cdot)$ is the zero function if $\ell$ is even.	 
	Motivated by equation (3.14) in \cite{EOP},
	we consider the $q$-bracket of the following generating function involving $Q_{\ell}(\cdot)$ with $\ell\in\mathbb{Z}_{>0}$,
	\begin{align}\label{eqn:def M}
		\begin{split}
		M(\mathbf{s})
		:=&\bigg\langle\exp\left(\sum\limits_{\ell=1}^\infty s_\ell(\ell-1)!Q_\ell\right)\bigg\rangle_q^s\\
		=&\exp\left(\sum\limits_{\ell=1}^\infty s_\ell(\ell-1)!\beta_\ell\right)\bigg\langle\exp\left(\sum\limits_{\ell=1}^\infty s_\ell P_{\ell-1}\right)\bigg\rangle_q^s.
		\end{split}
	\end{align}
	The coefficients of the Taylor expansion with respective to variables $(s_1,s_2,...)$ of the equation above give all possible $q$-bracket of products of some $Q_{\ell}(\cdot)$ with $\ell\in\mathbb{Z}_{>0}$.
	Moreover, notice that $Q_0(\lambda)=1$ for all self-conjugate partitions $\lambda$.
	Thus,
	this theorem is equivalent to the statement that,
	for a given sequence of positive integers $\mu=(\mu_1,...,\mu_{l(\mu)})$,
	the coefficient of $\prod_{j=1}^{l(\mu)} s_{2\mu_j}$ in the equation above is a quasimodular form of weight $2|\mu|:=2\sum_{j=1}^{l(\mu)}\mu_j$ for the congruence subgroup $\Gamma(2)$.
	Since the order of $\mu_i, 1\leq i\leq l(\mu)$ does not influence the result,
	we can assume $\mu_1\geq\cdots\geq\mu_{l(\mu)}$,
	and thus $\mu$ could be a partition.
 	
	From the definition \eqref{eqn:def P_l} of the function $P_\ell(\lambda), \ell\in\mathbb{Z}_{\geq0}$,
	the last term in the right hand side of equation \eqref{eqn:def M} can be computed as
	\begin{align}\label{eqn:expP}
		\begin{split}
			\bigg\langle&\exp\bigg(\sum\limits_{\ell=1}^\infty s_\ell P_{\ell-1}(\lambda)\bigg)\bigg\rangle_q^s
			\\
			&\ =\frac{1}{\prod_{j=0}^\infty(1+q^{2j+1})}
			\cdot\sum\limits_{\lambda\in\mathscr{P}^s}\prod_{i=1}^{r(\lambda)}q^{2m_i+1}\exp\bigg(\sum\limits_{\ell\in 2\Z_+} 2s_\ell (m_i+\frac{1}{2})^{\ell-1}\bigg).
		\end{split}			
	\end{align}
	From the one-to-one correspondence between partitions and their Frobenius coordinates,
	one can immediately find that the right hand side of equation \eqref{eqn:expP} is exactly equal to
	\begin{align*}
		\frac{1}{\prod_{j=0}^\infty(1+q^{2j+1})}\prod_{j=0}^\infty\Bigg(1+q^{2j+1}\cdot\exp\bigg(\sum\limits_{\ell\in 2\Z_+} 2s_\ell (j+\frac{1}{2})^{\ell-1}\bigg)\Bigg).
	\end{align*}
	With the computation result above, we take the logarithm of equation \eqref{eqn:def M} to obtain 
	\begin{align}\label{logQgen}
		\begin{split}
			\log M&(\mathbf{s})=\sum\limits_{\ell\in 2\Z_+} s_\ell(\ell-1)!\beta_\ell\\
			+&\sum_{j=0}^\infty\sum\limits_{n=1}^{\infty}(-1)^{n-1}\frac{q^{(2j+1)n}}{n}\cdot\Bigg[\exp\bigg(\sum\limits_{\ell\in 2\Z_+} 2ns_\ell (j+\frac{1}{2})^{\ell-1}\bigg)-1\Bigg].
		\end{split}	
	\end{align}	

	 From now on,
	 let $q=e^{\pi i \tau}$ and we assume the following expansion formula
	 \begin{align}\label{equ:logMs}
	 	\log M(\mathbf{s})
	 	=\sum_{\mu\in \mathscr{P}} M_\mu \cdot \prod_{i=1}^{l(\mu)}s_{2\mu_i}.
	 \end{align}
	 Then we have $M_\emptyset=0$ from the equation \eqref{logQgen}.
	 Below,
	 we are going to show that,
	 for any $\mu\neq\emptyset$,
	 $M_{\mu}$ is a quasimodular form of weight $2|\mu|=2\sum_{j=1}^{l(\mu)}\mu_j$ for the congruence subgroup $\Gamma(2)$.
	 
	 For the case of $l(\mu)=1$, 
	 we first recall that
	 \begin{align*}
	 	\sum\limits_{\ell=0}^\infty\beta_\ell x^\ell=&\frac{x/2}{\sinh(x/2)}=\sum\limits_{n=0}^\infty\frac{(1/2^{2n-1}-1)B_{2n}x^{2n}}{(2n)!},
	 \end{align*}
	then 
	 \begin{align*}
	 	(2n-1)!2^{2n-2}\beta_{2n}=\frac{(1-2^{2n-1})B_{2n}}{2\cdot2n}=-(1-2^{2n-1})\zeta(1-2n)/2.
	 \end{align*}
	 Thus,
	 from the equation \eqref{logQgen},
	 the $M_\mu$ for $\mu=(\mu_1)$ is given by
	 \begin{align}\label{eqn:eisenl=1}
	 	\begin{split}
	 	M_\mu=&(2\mu_1-1)!\beta_{2\mu_1}
	 	-2^{2-2\mu_1}\sum\limits_{n=1}^{\infty}\sum\limits_{\substack{d|n,2\nmid d\\d>0}}
	 	(-1)^{n}d^{2\mu_1-1}e^{\pi i n\tau}\\
	 	=&-2^{2-2\mu_1}\mathbb{G}_{2\mu_1}(\tau),
	 	\end{split}
	 \end{align}
	 where $\mathbb{G}_{2\mu_1}(\tau)$ is defined in equation \eqref{eqn:def bbG} and is a quasimodular form of weight $2\mu_1$ for $\Gamma(2)$.
	 
	 For the case of $l(\mu)>1$,
	 we denote $|\text{Aut}(\mu)|=\prod_{i\geq1}m_i(\mu)!$,
	 where $m_i(\mu)=\#\{j|\mu_j=i\}$.
	 Then still from the equation \eqref{logQgen},
	 we have
	\begin{align}\label{eqn:eisen}
		\begin{split}
			M_\mu
			=&\sum_{j=0}^\infty\sum\limits_{n=1}^{\infty}
			(-1)^{n-1}\frac{q^{(2j+1)n}}{n}
			\cdot\frac{1}{|\text{Aut}(\mu)|}2^{l(\mu)}n^{l(\mu)}(j+\frac{1}{2})^{2|\mu|-l(\mu)}\\
			=&-\frac{2^{2l(\mu)-2|\mu|}}{|\text{Aut}(\mu)|}
			\sum\limits_{n=1}^{\infty}\sum\limits_{d|n,2\nmid d\atop d>0}
			(-1)^{n}n^{l(\mu)-1}d^{2|\mu|-2l(\mu)+1}e^{\pi i n\tau}\\
			=&-\frac{2^{2l(\mu)-2|\mu|}}{|\text{Aut}(\mu)|}\bigg(\frac{1}{\pi i}\frac{\partial}{\partial \tau}\bigg)^{l(\mu)-1}\mathbb{G}_{2|\mu|-2l(\mu)+2}(\tau),
		\end{split}
	\end{align}
	which is a quasimodular form of weight
	\begin{align*}
		2|\mu|-2l(\mu)+2+2\big(l(\mu)-1\big)
		=2|\mu|
	\end{align*}
	since the operator $\frac{\partial}{\partial \tau}$ preserves the space of quasi-modular forms and increases the weight by $2$.
	
	In conclusion, for any partition $\mu$,
	$M_\mu$ is a quasimodular form of weight $2|\mu|$ for $\Gamma(2)$.
	As a consequence,
	by taking exponentiation of equation \eqref{logQgen},
	the coefficient of $\prod_{i=1}^{l(\mu)}s_{2\mu_i}$ in $M(\mathbf{s})$ is also a quasimodular form of weight $2|\mu|$ for $\Gamma(2)$,
	since the space of quasimodular forms for $\Gamma(2)$ is a graded ring.
	Thus, the proof of this theorem is finished.
\end{proof}

\begin{rmk}\label{rek:onetwo}
	For the case of $n=1$ and $n=2$, the quasimodularity can also be directly derived from the explicit formulas for the one-point and two-point functions exhibited in Corollary \ref{cor:main one}. More precisely, by considering the expansion of equation \eqref{eqn:G as Gl}, $\langle Q_k\rangle ^s_q$ is an element in the ring of quasimodular forms. 
	For the case of $n=2$,
	one can show (similar to the method used in proving Proposition \ref{1dmodular})
	\begin{align}\label{equ:dtheta1}
		\begin{split}
			\frac{\Theta_1^{'}(t;q)}{\Theta_1(t;q)}=t\frac{\partial}{\partial t}\log\Theta_1(t;q)
			&=\frac{1}{2\pi iz}-\sum\limits_{\ell\in 2\Z_+} 2G_\ell(\tau)\frac{(2\pi iz)^{\ell-1}}{(\ell-1)!}
		\end{split}
	\end{align}
	for $\Theta_1(t;q)$, and
	\begin{align}\label{equ:dtheta2}
		\begin{split}
			\frac{\Theta_3^{'}(t;q)}{\Theta_3(t;q)}=t\frac{\partial}{\partial t}\log\Theta_3(t;q)
			&=-\sum\limits_{\ell\in 2\Z_+} 2G_\ell^{(1,1)}(\tau)\frac{(2\pi iz)^{\ell-1}}{(\ell-1)!}
		\end{split}
	\end{align}
	for $\Theta_3(t;q)$.
	Then the quasimodularity of this case directly follows from the formula \eqref{eqn:G2 main} for the two-point function
	\begin{align*}
		G(t_1,t_2)=G(t_1t_2)\cdot 
		\frac{\Theta_3^{'}(t_1t_2;q)}{\Theta_3(t_1t_2;q)}
		-G(t_1/t_2)\cdot \frac{\Theta_1^{'}(t_1;q)}{\Theta_1(t_1;q)}
		-G(t_2/t_1)\cdot \frac{\Theta_1^{'}(t_2;q)}{\Theta_1(t_2;q)},
	\end{align*}
	together with the equations \eqref{eqn:G as Gl}, \eqref{equ:dtheta1} and \eqref{equ:dtheta2}.
\end{rmk}

\section{Limit shape of the self-conjugate partitions under Gibbs uniform measure}
\label{sec:limit shape}
In this section,
we derive the limit shape of the self-conjugate partitions under the measure $\mathfrak{M}_q(\cdot)$ when $q\rightarrow1^-$ and verify its compatibility with the leading asymptotics of the one-point function $G(t)$.

\subsection{Limit shape of the self-conjugate partitions under the measure $\mathfrak{M}_q(\cdot)$ when $q\rightarrow1^-$}
In this subsection,
we study the limit shape of the self-conjugate partitions under the measure $\mathfrak{M}_q(\cdot)$ when $q\rightarrow1^-$ and prove Proposition \ref{prop:limit shape}. We mainly follow the method in \cite{FVY}.

We first derive the typical size of self-conjugate partitions,
which indicates how to rescale the limit Young diagrams.
Throughout this section, we shall apply the substitution $q=e^{-2\pi r}$.
\begin{lem}\label{lem:typical size}
	The typical size of the self-conjugate partitions under the measure $\mathfrak{M}_q(\cdot)$ when $q\rightarrow1^-$,
	or equivalently $r\rightarrow 0^+$, is given by
	\begin{align*}
		\lim_{q\rightarrow1^-}
		r^2 \cdot \mathbb{E}_q(|\cdot|)
		=\frac{1}{96},
	\end{align*}
	where $|\cdot|$ represents the size function on the set of self-conjugate partitions.
\end{lem}
\begin{proof}
Recall that the generating function of self-conjugate partitions is of the following form,
\begin{align*}
	Z_s(q)
	=\sum\limits_{\lambda\in\mathscr{P}^s}q^{|\lambda|}
	=\prod_{k=0}^\infty
	\big(1+q^{2k+1}\big).
\end{align*}
Then,
the expectation value of the size of the self-conjugate partitions is given by
\begin{align*}
	\mathbb{E}_q(|\cdot|)
	=&q\frac{\partial}{\partial q}
	Z_s(q) \big/Z_s(q)
	=\sum_{k=0}^\infty
	\frac{(2k+1)q^{2k+1}}{1+q^{2k+1}}.
\end{align*}
By using the substitution $q=e^{-2\pi r}$,
we have
\begin{align}
	\mathbb{E}_q(|\cdot|)
	=&\frac{1}{r}
	\sum_{k=0}^\infty
	\frac{r(2k+1)e^{-2\pi(2k+1)r}}{1+e^{-2\pi(2k+1)r}}\\
	=&\frac{1}{r}
	\sum_{k=0}^{\lfloor\frac{M}{r}\rfloor}
	\frac{r(2k+1)e^{-2\pi(2k+1)r}}{1+e^{-2\pi(2k+1)r}}
	+\frac{1}{r}
	\sum_{k=\lfloor\frac{M}{r}\rfloor+1}^{\infty}
	\frac{r(2k+1)e^{-2\pi(2k+1)r}}{1+e^{-2\pi(2k+1)r}} \label{eqn:Eq||2}
\end{align}
for a fixed sufficient large number $M$. 
For instance, one can take $M=1/r$.
It is obvious that the first part of the equation \eqref{eqn:Eq||2} goes to a Riemann sum and the second part goes to $0$ when $r \rightarrow 0^+$.
Thus,
as $q\rightarrow1^-$,
\begin{align*}
	r^2 \cdot \mathbb{E}_q(|\cdot|)
	\rightarrow&
	\frac{1}{2}
	\int_0^{2(\lfloor1/r^2 \rfloor+1)r} \frac{xe^{-2\pi x}}{1+e^{-2\pi x}} dx\\
	\rightarrow&
	\frac{1}{2}\cdot\int_0^{\infty} \frac{xe^{-2\pi x}}{1+e^{-2\pi x}} dx
	=\frac{1}{96}.
\end{align*}
\end{proof}

For a partition $\lambda\in\mathscr{P}^s$,
define
$$m_k(\lambda):=\#\{i|\ \lambda_i=k\},\qquad 0\leq k \leq \lambda_1,$$
which is the number of $k$ appearing in $\lambda$.
To describe the limit shape of self-conjugate partitions,
we introduce the following function
\begin{align*}
	f_\lambda(x)
	:=-\sum_{k\geq x} m_{k}(\lambda).
\end{align*}
Intuitively,
the graph of this function $f_\lambda(x)$ is exactly the lower boundary of the Young diagram corresponding to the partition $\lambda$.
See the Figure \ref{fig:Yd 53211 and 41} as an example of $f_\lambda(x)$ for $\lambda=(5,3,2,1,1)$.

\begin{figure}[htbp]
	\begin{tikzpicture}[scale=0.7]
		\draw[help lines, color=gray!30] (-1,-6) grid (6,1);
		\draw[thick,->](-1,0)--(6,0);
		\draw[thick,->](0,-6)--(0,1);
		%%%The Young Diagram
		\draw[thick](0,-5)--(1,-5);
		\draw[thick](1,-3)--(1,-5);
		\draw[thick](1,-3)--(2,-3);
		\draw[thick](2,-2)--(2,-3);
		\draw[thick](2,-2)--(3,-2);
		\draw[thick](3,-1)--(3,-2);
		\draw[thick](3,-1)--(5,-1);
		\draw[thick](5,0)--(5,-1);
		\draw[fill,thick] (3,-2.5) node {$f_\lambda(x)$};
	\end{tikzpicture}
	\caption{The graph of $f_\lambda(x)$ for $\lambda=(5,3,2,1,1)$}
	\label{fig:Yd 53211 and 41}
\end{figure}
\noindent For convenience, we can regard $f_{.}(x)$ as a function on $\mathscr{P}^s$ depending on $x$.

Indicated by Lemma \ref{lem:typical size},
we introduce the rescaled function $\tilde{f}_{\lambda}(x)$ by
\begin{align*}
	\tilde{f}_\lambda(x):=
	4\sqrt{6}r \cdot f_\lambda(x/4\sqrt{6}r),
\end{align*}
where $r=-\frac{1}{2\pi}\cdot \log q$.
The goal of this subsection is to study the limit behavior of $\tilde{f}_\lambda(x)$ under the measure $\mathfrak{M}_q(\cdot)$ when $q\rightarrow1^-$ and prove the Proposition \ref{prop:limit shape}.
For convenience,
we restate it as follows.
\begin{prop}[=Proposition \ref{prop:limit shape}]
	For any fixed $x>0$ and $\epsilon>0$,
	we have the following limit
	\begin{align*}
		\lim_{q\rightarrow1^-}
		\mathfrak{M}_q\big(\big\{\lambda\big|\ |\tilde{f}_{\lambda}(x)-f(x)|<\epsilon\big\}\big)=1,
	\end{align*}
	where $f(x)=\frac{\sqrt{6}}{\pi}\log\big(1-\exp(-\pi x/\sqrt{6})\big)$
	is already introduced in equation \eqref{eqn:limit shape main}.
\end{prop}
\begin{proof}
When studying the self-conjugate partitions,
it is much easier to consider the following
$$\alpha_k(\lambda):=\#\{i|\ \lambda_i-i=k\},\ \ \ \  0\leq k \leq r(\lambda),$$
and the function
\begin{align*}
	g_\lambda(x)
	:=-\sum_{k\geq x} \alpha_{k}(\lambda)
\end{align*}
instead of $m_k(\lambda)$ and $f_{\lambda}(x)$.
Hence we can regard $g_{.}(x)$ as a function on $\mathscr{P}^s$ depending on $x$.
See the Figure \ref{fig:Yd 41} as an example of $g_\lambda(x)$ for $\lambda=(5,3,2,1,1)$.
\begin{figure}[htbp]
	\begin{tikzpicture}[scale=0.7]
		\draw[help lines, color=gray!30] (-1,-4) grid (6,1);
		\draw[thick,->](-1,0)--(6,0);
		\draw[thick,->](0,-4)--(0,1);
		%%%The Young Diagram
		\draw[thick](0,-2)--(1,-2);
		\draw[thick](1,-1)--(1,-2);
		\draw[thick](1,-1)--(4,-1);
		\draw[thick](4,0)--(4,-1);
		\draw[fill,thick] (2.5,-1.5) node {$g_\lambda(x)$};
	\end{tikzpicture}
	\caption{The graph of $g_{\lambda}(x) $ for $\lambda=(5,3,2,1,1)$}
	\label{fig:Yd 41}
\end{figure}
Similar to the $\tilde{f}_\lambda(x)$,
when studying the limit $q\rightarrow1^-$,
we should consider the rescaled $g_{\lambda}(x)$ defined as
\begin{align}\label{eqn:glambda}
	\tilde{g}_\lambda(x):=
	4\sqrt{6}r \cdot g_\lambda(x/4\sqrt{6}r).
\end{align}

Recall that $r(\lambda)$ is the Frobenius length of the partition $\lambda$.
For $1\leq k \leq r(\lambda)$,
these two quantities $\alpha_k(\lambda)$ and $m_k(\lambda)$ can transform themselves into each other.
Thus we can actually use $g_{\lambda}(x)$ to study $f_{\lambda}(x)$.
Denote by $\overline{g}_{\lambda}(x)$ and $\overline{f}_{\lambda}(x)$ the functions obtained by rotating $g_{\lambda}(x)$ and $f_{\lambda}(x)$ 90 degree counterclockwise, respectively.
Then one has
\begin{align}\label{eqn:ff and gg}
	\overline{f}_{\lambda}(x)
	=\overline{g}_{\lambda}(x)+\lfloor x \rfloor+1
\end{align}
for $0\leq x \leq r(\lambda)$.
See Figure \ref{fig:Yd 53211 and 41} as an example for the partition $(5,3,2,1,1)$.
Moreover,
using the self-conjugate property of $\lambda$,
we can recover the whole $\overline{f}_{\lambda}(x)$ by rotating the graph of $\overline{f}(x), 0\leq x \leq r(\lambda)$ over the line $y=x$. 
As a consequence, we shall derive the limit behavior for $\tilde{g}_\lambda(x)$ first and recover the result of $\tilde{f}_\lambda(x)$  through $\tilde{g}_\lambda(x)$.

We first calculate the limit Frobenius length of self-conjugate partitions under the measure $\mathfrak{M}_q(\cdot)$ when $q\rightarrow1^-$,
which enables us to locate the interval where we can use $\tilde{g}_{\lambda}(x)$ to study $\tilde{f}_{\lambda}(x)$ directly.
By virtue of the generating function
\begin{align*}
	Z_s(b,q):=\sum_{\lambda\in\mathcal{P}^s}
	b^{r(\lambda)} q^{|\lambda|}
	=\prod_{i=0}^\infty \big(1+b q^{2i+1}\big),
\end{align*}
the expectation value of $r(\cdot)$ is given by
\begin{align*}
	\mathbb{E}_q\big(r(\cdot)\big)
	=\bigg(b\frac{\partial}{\partial b}Z_s(b,q) \big/Z_s(b,q)\bigg)\Big|_{b=1}
	=\sum_{i=0}^\infty \frac{q^{2i+1}}{1+q^{2i+1}}.
\end{align*}
Moreover, by the substitution $q=e^{-2\pi r}$,
\begin{align}
	4\sqrt{6} r \cdot \mathbb{E}_q\big(r(\cdot)\big)
	=&4\sqrt{6}r
	\sum_{i=0}^\infty
	\frac{e^{-2\pi(2i+1)r}}{1+e^{-2\pi(2i+1)r}}\\
	=&4\sqrt{6}r
	\sum_{i=0}^{\lfloor M/r \rfloor}
	\frac{e^{-2\pi(2i+1)r}}{1+e^{-2\pi(2i+1)r}}
	+4\sqrt{6}r
	\sum_{i=\lfloor M/r \rfloor+1}^\infty
	\frac{e^{-2\pi(2i+1)r}}{1+e^{-2\pi(2i+1)r}} \label{eqn:E of r ar two parts}
\end{align}
for any positive number $M$.
Here,
we take $M=1/r$.
Then,
for the second part in equation \eqref{eqn:E of r ar two parts},
we have
\begin{align*}
	4\sqrt{6}r
	\sum_{i=\lfloor M/r \rfloor+1}^\infty
	\frac{e^{-2\pi(2i+1)r}}{1+e^{-2\pi(2i+1)r}}
	\leq
	\frac{4\sqrt{6}r \cdot e^{-2\pi(2/r^2+1)r}}{1-e^{-4\pi r}}
	\leq \frac{2\sqrt{6} \cdot e^{-4\pi/r}}{\pi}
\end{align*}
as $r \rightarrow 0^+$. 
About the first part in equation \eqref{eqn:E of r ar two parts},
it is a Riemann sum for the following integral
\begin{align*}
	2\sqrt{6}
	\cdot \int_{0}^{ 2(\lfloor1/r^2 \rfloor+1)r}
	\frac{e^{-2\pi s}}{1+e^{-2\pi s}} ds
	=&2\sqrt{6}
	\cdot\frac{1}{2\pi}\int_{0}^{1}
	\frac{1}{1+s} ds
	+O(e^{-4\pi/r})\\
	=&\frac{\sqrt{6} \log 2}{\pi}
	+O(e^{-4\pi/r}).
\end{align*}
As a consequence,
we have
\begin{align}\label{eqn:typical F length}
	4\sqrt{6} r \cdot \mathbb{E}_q\big(r(\cdot)\big)
	=\frac{\sqrt{6} \log 2}{\pi}
	+O(r),
\end{align}
as $r\rightarrow0^+$.

On the other hand,
about the variance of the Frobenius length $r(\cdot)$,
we need to compute
\begin{align*}
	\mathbb{E}_q\big(r(\cdot)^2\big)
	=&\Bigg(\bigg(b\frac{\partial}{\partial b}\bigg)^2 Z_s(b,q) \big/Z_s(b,q)\Bigg)\Big|_{b=1}\\
	=&\mathbb{E}_q\big(r(\cdot)\big)^2
	+\sum_{i=0}^\infty \frac{q^{2i+1}}{(1+q^{2i+1})^2}.
\end{align*}
It is obvious that $\sum\limits_{i=0}^\infty \frac{q^{2i+1}}{(1+q^{2i+1})^2}=O(1/r)$.
Thus,
\begin{align*}
	&\mathfrak{M}_q\big(\big\{\lambda\big|\ |4\sqrt{6} r \cdot r(\lambda)
	-\sqrt{6}\log2/\pi|>\epsilon\big\}\big)\\
	&\quad\quad
	\leq 
	\mathbb{E}_q\big((4\sqrt{6} r \cdot r(\lambda)
	-\sqrt{6}\log2/\pi)^2\big) \cdot 1/\epsilon^2
	\leq 1/\epsilon^2 \cdot O(r),
\end{align*}
which goes to $0$ as $r\rightarrow 0^+$,
i.e., $q\rightarrow1^-$.
That is to say,
the limit rescaled Frobenius length $4\sqrt{6} r \cdot r(\cdot)$ of self-conjugate partitions, under the measure $\mathfrak{M}_q(\cdot)$ when $q\rightarrow1^-$, is $\sqrt{6} \log 2/\pi$.

From now on,
we study the function $\tilde{g}_{\lambda}(x)$ defined in the equation \eqref{eqn:glambda}.
The following probabilities are needed in our computations,
\begin{align*}
	\mathfrak{M}_q\big(\alpha_k(\cdot)=1\big)
	=\frac{q^{2k+1}}{1+q^{2k+1}},\qquad
	\mathfrak{M}_q\big(\alpha_k(\cdot)=0\big)
	=\frac{1}{1+q^{2k+1}}.
\end{align*}
Thus,
from the definition of $g_\lambda(x)$,
\begin{align*}
	\mathbb{E}_q \big( g_\cdot(x)\big)
	=-\sum_{k\geq x} \mathbb{E} \big(\alpha_{k}(\cdot)\big)
	=-\sum_{k\geq x} \frac{q^{2k+1}}{1+q^{2k+1}}.
\end{align*}
Then, the expectation value of $\tilde{g}_\cdot(x)$ is given by
\begin{align*}
	\mathbb{E}_q \big(\tilde{g}_\cdot(x)\big)
	=&-4\sqrt{6}r
	\sum_{k\geq x/4\sqrt{6}r}
	\frac{e^{-2\pi(2k+1)r}}{1+e^{-2\pi(2k+1)r}},
	\end{align*}
	which is a Riemann sum.
	Thus, the equation above is equal to
	\begin{align*}
	-2\sqrt{6}\cdot
	\int_{ x/2\sqrt{6}}^\infty
	\frac{e^{-2\pi t}}{1+e^{-2\pi t}} dt
	+O(r)
	=&-2\sqrt{6}
	\cdot\frac{1}{2\pi}\int_{0}^{e^{-\pi x/\sqrt{6}}}
	\frac{1}{1+s} ds
	+O(r)\\
	=&-\frac{\sqrt{6}}{\pi}
	\log(1+e^{-\pi x/\sqrt{6}})
	+O(r).
\end{align*}
By the similar method in analyzing the limit rescaled Frobenius length,
we can obtain the limit behavior of $\tilde{g}_{\lambda}(t)$.
More precisely,
denote
\begin{align}\label{eqn:def g}
	g(x):=\mathbb{E}_q \big(\tilde{g}_\cdot(x)\big)=-\frac{\sqrt{6}}{\pi}
	\log(1+e^{-2\pi x/2\sqrt{6}}).
\end{align}
We have,
for any fixed $x>0$ and $\epsilon>0$,
\begin{align*}
	\lim_{q\rightarrow1^-}
	\mathfrak{M}_q\big(\big\{\lambda\big|\ |\tilde{g}_{\lambda}(x)-g(x)|<\epsilon\big\}\big)=1.
\end{align*}

Now,
we use the relation between $\tilde{f}_\lambda(x)$ and $\tilde{g}_{\lambda}(x)$ to derive the limit behavior of $\tilde{f}_\lambda(x)$.
First,
recall that $\overline{f}_{\lambda}(x)$ and $\overline{g}_{\lambda}(x)$ are obtained from $f_{\lambda}(x)$ and $g_{\lambda}(x)$ by rotating 90 degree counterclockwise, respectively.
We denote $\overline{f}(x)$ and $\overline{g}(x)$ as the limit rescaled $\overline{f}_{\cdot}(x)$ and $\overline{g}_{\cdot}(x)$ under the measure $\mathfrak{M}_q(\cdot)$ when $q\rightarrow1^-$.
then the relation \eqref{eqn:ff and gg} between $\overline{f}_{\lambda}(x)$ and $\overline{g}_{\lambda}(x)$ is reduced to the following
\begin{align}\label{eqn:ff and gg2}
	\overline{f}(x)
	=\overline{g}(x)+x.
\end{align}
After rotating $90$ degree counterclockwise,
the function $g(x)$ defined in the equation \eqref{eqn:def g} becomes
\begin{align*}
	\overline{g}(x)=-\frac{\sqrt{6}}{\pi}\log\big(-1+\exp(\pi x/\sqrt{6})\big).
\end{align*}
The region, in which $\overline{g}(x)$ is related to $\overline{f}(x)$ in terms of equation \eqref{eqn:ff and gg2}, is given by $0<x\leq\frac{\sqrt{6}\log2}{\pi}$.
The upper bound $\frac{\sqrt{6}\log2}{\pi}$ is the limit rescaled Frobenius length of self-conjugate partitions given in equation \eqref{eqn:typical F length}.
One can also verify that the zero of the function $\overline{g}(x)$ is exactly $x_0=\frac{\sqrt{6}\log2}{\pi}$,
which is compatible with the limit rescaled Frobenius length.

Second,
by the relation \eqref{eqn:ff and gg2},
the function $\overline{f}(x)$ is then obtained as
\begin{align}\label{eqn:def ff}
	\overline{f}(x)=x+\overline{g}(x)
	=x-\frac{\sqrt{6}}{\pi}\log\big(-1+\exp(\pi x/\sqrt{6})\big)
\end{align}
for $0<x\leq\frac{\sqrt{6}\log2}{\pi}$.

Last, after rotating the graph of $\overline{f}(x)$ in equation \eqref{eqn:def ff} $90$ degree clockwise,
we obtain the limit rescaled graph $\tilde{f}_{\cdot}(x)$ of self-conjugate partitions under the measure $\mathfrak{M}_q(\cdot)$ when $q\rightarrow1^-$,
\begin{align}\label{eqn:f final}
	f(x)=\frac{\sqrt{6}}{\pi}\log\big(1-\exp(-\pi x/\sqrt{6})\big),
\end{align}
which is valid in the region $-\frac{\sqrt{6}\log2}{\pi}\leq f(x) <0$.
Notice that,
the graph of the function $f(x)$ in equation \eqref{eqn:f final} is invariant under rotation around the line $y=-x$.
Thus,
the graph of $f(t)$ is exactly the limit rescaled Young diagram of self-conjugate partitions in the whole region and this proposition is then proved.
\end{proof}

\begin{rmk}
	The limit shape given in equation \eqref{eqn:limit shape main} is equivalent to the limit shape of large integer partitions under the uniform measure derived in \cite{V96} (see also \cite{O01}) after rotation,
	even though the set of self-conjugate partitions is a very small part of the set of all integer partitions.
\end{rmk}

\subsection{Comparison of the leading asymptotics of the one-point function and the limit shape}
In this subsection,
we show that the leading asymptotics of the one-point function $G(t)$ matches the limit shape derived in the last subsection.

We use the notation $\Lambda$ to denote a partition depending on $q$,
which has the limit shape given by $f(x)$ in equation \eqref{eqn:limit shape main}.
It is indeed a typical partition, without scaling, under the measure  $\mathfrak{M}_q(\cdot)$ when $q\rightarrow1^-$.
More precisely,
for any fixed $q=e^{-2\pi r}$,
we pick the partition $\Lambda$ given by
\begin{align}\label{eqn:def Lambda}
	\begin{split}
	\Lambda_i
	:=&-\Big\lfloor\frac{1}{4\sqrt{6}r} f\big(4\sqrt{6}r x\big)|_{x=i}
	\Big\rfloor\\
	=&-\Big\lfloor\frac{1}{4\pi r} \log\big(1-\exp(-4\pi ri)\big)
	\Big\rfloor,
	\quad i=1,2,....
	\end{split}
\end{align}
When considering the one-point function $G(t)$,
we mainly study the function $\mathcal{T}(\cdot)$ on the set of self-conjugate partitions,
which is defined by
\begin{align*}
	\mathcal{T}(\lambda):=\sum_{i=1}^\infty t^{\lambda_i-i+1/2},
\end{align*}
for a partition $\lambda \in \mathscr{P}^s$. Thus,
we can use the notation $\mathcal{T}(\Lambda)$ to denote the action of this function $\mathcal{T}(\cdot)$ on the typical partition $\Lambda$.
The main result in this subsection is stated as follows.
\begin{cor}\label{cor:aysm}
The leading asymptotics of the $\mathcal{T}(\Lambda)$ and the one-point function $G(t)$ are the same,
i.e.,
\begin{align*}
	\lim_{q\rightarrow1^-}
	\tau \cdot G(t)|_{z\rightarrow \tau z}
	=\lim_{q\rightarrow 1^-} \tau\cdot \mathcal{T}(\Lambda)|_{z\rightarrow \tau z},
\end{align*}
where $q=e^{\pi i \tau}$ and $t=e^{2\pi i z}$.
\end{cor}
\begin{proof}
This corollary should immediately follow from Proposition \ref{prop:limit shape} since $\Lambda$ is the typical partition of the measure $\mathfrak{M}_q(\cdot)$ when $q\rightarrow1^-$.
Here, we provide a direct proof of this result.

From the definition \eqref{eqn:def Lambda} of the typical partition $\Lambda$,
the value $\mathcal{T}(\Lambda)$ is given by
\begin{align*}
	\mathcal{T}(\Lambda)
	=\sum_{i=1}^\infty t^{-\big\lfloor\frac{1}{4\pi r} \log\big(1-\exp(-4\pi ri)\big)
		\big\rfloor -i+1/2},
\end{align*}
where we use the notation $q=e^{\pi i \tau}=e^{-2\pi r}$. Since $\tau=\frac{1}{\pi i}\log q=2ir$ and $t=e^{2\pi i z}$,
\begin{align*}
	\tau\cdot \mathcal{T}(\Lambda)|_{z\rightarrow \tau z}
	=2i\cdot r
	\cdot \sum_{i=1}^\infty e^{-4\pi z \cdot r\big(-\big\lfloor\frac{1}{4\pi r} \log\big(1-\exp(-4\pi ri)\big)
		\big\rfloor -i+1/2\big)}.
\end{align*}
As $r\rightarrow0^+$,
the summation above is convergent to a Riemann sum.
So we have
\begin{align*}
	\lim_{q\rightarrow 1^-} \tau\cdot \mathcal{T}(\Lambda)|_{z\rightarrow \tau z}
	=&2i
	\cdot \int_{0}^\infty e^{z\log(1-\exp(-4\pi x))+4\pi z x}dx\\
	=&2i
	\cdot \int_{0}^\infty \big(\exp(4\pi x)-1\big)^zdx
	=\frac{1}{2\text{sinh}(\pi i z)}.
\end{align*}

On the other hand, 
to obtain the leading asymptotics of the one-point function $G(t)$,
we need to know the asymptotic behaviors of the Eisenstein series $G_\ell(\tau)$ and $G^{(1,1)}_\ell(\tau)$.
Actually,
from the following equivalent definitions of $G_\ell(\tau)$ and $G^{(1,1)}_\ell(\tau)$ (see, for example, \cite{DS}),
\begin{align*}
	&G_\ell(\tau)=\frac{(\ell-1)!}{(2\pi i)^\ell}\sum_{(c,d)\in \N^2\atop(c,d)\neq (0,0)}\frac{1}{(c\tau+d)^\ell},\\
	&G^{(1,1)}_\ell(\tau)=\frac{(\ell-1)!}{(\pi i)^\ell}\sum_{(c,d)\in \N^2\atop(c,d)\equiv (1,1)\mod 2}\frac{1}{(c\tau+d)^\ell},
\end{align*}
we have
\begin{align*}
	&\lim_{\tau\rightarrow0} \tau^\ell G_\ell(\tau) 
	=\frac{(\ell-1)!}{(2\pi i)^\ell}\zeta(\ell)=-\frac{B_\ell}{2\ell},\\
	&\lim_{\tau\rightarrow0} \tau^\ell G^{(1,1)}_\ell(\tau) 
	=0.
\end{align*}
Thus,
from the explicit formula \eqref{eqn:G as Gl} for the one-point function $G(t)$,
we have
\begin{align}
	\lim_{q\rightarrow1^-}
	\tau \cdot G(t)|_{z\rightarrow \tau z}
	&=\frac{1}{2\pi i z}\exp\bigg(-\sum\limits_{\ell\in2\Z_+}
	\frac{B_\ell}{\ell}
	\frac{(2\pi i z)^\ell}{\ell!}\bigg)\\
	&=\frac{1}{e^{\pi i z}-e^{-\pi i z}}=\frac{1}{2\text{sinh}(\pi i z)},
\end{align}
where in the second equal sign,
we have applied the identity \eqref{eqn:Bernoulli}.
This finishes the proof.
\end{proof}

\begin{ex}
	We list a few of terms of the leading asymptotics of the one-point function $G(t)$ as,
\begin{align*}
	&\lim_{q\rightarrow1^-}
	\tau \cdot G(t)|_{z\rightarrow \tau z}
	=\lim_{q\rightarrow 1^-} \tau\cdot \mathcal{T}(\Lambda)|_{z\rightarrow \tau z}\\                                 
	=&-\frac{i}{2 \pi  z}-\frac{i \pi }{12}  z-\frac{7i \pi ^3}{720}  z^3-\frac{31 i \pi ^5 }{30240}z^5
	-\frac{127 i \pi ^7 }{1209600}z^7
	-\frac{73 i \pi ^9 }{6842880}z^9\\
	&-\frac{1414477 i \pi ^{11} }{1307674368000}z^{11}
	-\frac{8191 i \pi ^{13} }{74724249600}z^{13}
	-\frac{16931177 i \pi ^{15} }{1524374691840000}z^{15}\\
	&-\frac{5749691557 i \pi ^{17} }{5109094217170944000}z^{17}
	-\frac{91546277357 i \pi ^{19} }{802857662698291200000}z^{19}
	+O\left(z^{21}\right).
\end{align*} 
\end{ex}

\section{Conflict of interest and data availability statement}
The authors state that there is no conflict of interest, and
no datasets were generated or analysed during the current study.

\setcounter{section}{0}
\setcounter{tocdepth}{2}
	
%\tableofcontents

\vspace{.2in}
{\em Acknowledgements}.
The authors are grateful to the anonymous referee for many excellent suggestions.
The authors would like to thank Professor Shuai Guo for his encouragement.
The second-named author would like to thank Professor Di Yang for helpful discussions.
The second-named author is supported by the NSFC grant (No. 12401079).
\vspace{.2in}

%%%%%%%%%%%%%%%%%%%%%%%%%%%%%%%%%%%%%%%%%%%%%%%
\renewcommand{\refname}{Reference}
\bibliographystyle{plain}
\bibliography{reference}
\vspace{30pt} \noindent
%%%%%%%%%%%%%%%%%%%%%%%%%%%%%%%%%%%%%%%%%%%%%%%%%%%%%%%%%%%%%%%%%%%%%%%%%%%%%%%%%%%
\end{document}